\documentclass[prx,twocolumn,floatfix,superscriptaddress,longbibliography,notitlepage]{revtex4-1}

\usepackage{graphicx, color, graphpap}
\usepackage{enumitem}
\usepackage{amssymb}
\usepackage{amsthm}
\usepackage{dsfont}
\usepackage{mathtools}
\usepackage[dvipsnames]{xcolor}
\usepackage{multirow}
\usepackage[colorlinks=true,citecolor=blue,linkcolor=magenta]{hyperref}
\usepackage[T1]{fontenc}
\usepackage{bbm}
\usepackage{thmtools,thm-restate}
\usepackage{verbatim}
\usepackage{mathtools}
\usepackage{titlesec}
\usepackage{amsmath}

\usepackage[caption = false]{subfig}

\long\def\ca#1\cb{} 

\newcommand{\ket}[1]{|#1\rangle}               
\newcommand{\bra}[1]{\langle #1|}              
\newcommand{\dya}[1]{\ket{#1}\!\bra{#1}}
\newcommand{\matl}[3]{\langle #1|#2|#3\rangle} 

\newcommand{\poly}{\operatorname{poly}}

\newcommand{\AC}{\mathcal{A}}
\newcommand{\BC}{\mathcal{B}}

\newcommand{\DC}{\mathcal{D}}

\newcommand{\NC}{\mathcal{N}}

\newcommand{\PC}{\mathcal{P}}

\newcommand{\VC}{\mathcal{V}}
\newcommand{\WC}{\mathcal{W}}

\newcommand{\Tr}{{\rm Tr}}

\renewcommand{\geq}{\geqslant}
\renewcommand{\leq}{\leqslant}

\renewcommand{\vec}[1]{\boldsymbol{#1}}  

\newcommand{\ad}{^\dagger}

\newcommand*{\id}{\openone}

\newcommand{\thv}{\vec{\theta}}

{}
{}

\newtheorem{theorem}{Theorem}
\newtheorem{lemma}{Lemma}
\newtheorem{suplemma}{Supplementary Lemma}
\newtheorem{corollary}{Corollary}
\newtheorem{proposition}{Proposition}

\newtheorem{remark}{Remark}

\begin{document}

\title{Noise-induced barren plateaus in variational quantum algorithms}

\author{Samson Wang}
\email{samsonwang@outlook.com}
\affiliation{Theoretical Division, Los Alamos National Laboratory, Los Alamos, NM 87545, USA}
\affiliation{Imperial College London, London, UK}

\author{Enrico Fontana}
\affiliation{Theoretical Division, Los Alamos National Laboratory, Los Alamos, NM 87545, USA}
\affiliation{University of Strathclyde, Glasgow, UK}
\affiliation{National Physical Laboratory, Teddington, UK}

\author{M.~Cerezo}
\email{cerezo@lanl.gov}
\affiliation{Theoretical Division, Los Alamos National Laboratory, Los Alamos, NM 87545, USA}
\affiliation{Center for Nonlinear Studies, Los Alamos National Laboratory, Los Alamos, NM, USA
}

\author{Kunal~Sharma} 
\address{Theoretical Division, Los Alamos National Laboratory, Los Alamos, NM 87545, USA}
\address{Hearne Institute for Theoretical Physics and Department of Physics and Astronomy, Louisiana State University, Baton Rouge, LA USA}

\author{Akira Sone}
\affiliation{Theoretical Division, Los Alamos National Laboratory, Los Alamos, NM 87545, USA}
\affiliation{Center for Nonlinear Studies, Los Alamos National Laboratory, Los Alamos, NM, USA
}

\author{Lukasz Cincio}
\affiliation{Theoretical Division, Los Alamos National Laboratory, Los Alamos, NM 87545, USA}

\author{Patrick J. Coles}
\email{pcoles@lanl.gov}
\affiliation{Theoretical Division, Los Alamos National Laboratory, Los Alamos, NM 87545, USA}

\begin{abstract}
Variational Quantum Algorithms (VQAs) may be a path to quantum advantage on Noisy Intermediate-Scale Quantum (NISQ) computers. A natural question is whether noise on NISQ devices places fundamental limitations on VQA performance. We rigorously prove a serious limitation for noisy VQAs, in that the noise causes the training landscape to have a barren plateau (i.e., vanishing gradient). Specifically, for the local Pauli noise considered, we prove that the gradient vanishes exponentially in the number of qubits $n$ if the depth of the ansatz grows linearly with $n$. These noise-induced barren plateaus (NIBPs) are conceptually different from noise-free barren plateaus, which are linked to random parameter initialization. Our result is formulated for a generic ansatz that includes as special cases the Quantum Alternating Operator Ansatz and the Unitary Coupled Cluster Ansatz, among others. For the former, our numerical heuristics demonstrate the NIBP phenomenon for a realistic hardware noise model.
\end{abstract}

\maketitle

\section{Introduction}

One of the great unanswered technological questions is whether Noisy Intermediate Scale Quantum (NISQ) computers will yield a quantum advantage for tasks of practical interest~\cite{preskill2018quantum}. At the heart of this discussion are Variational Quantum Algorithms (VQAs), which are believed to be the best hope for near-term quantum advantage~\cite{cerezo2020variationalreview,endo2020hybrid,bharti2021noisy}. Such algorithms leverage classical optimizers to train the parameters in a quantum circuit, while employing a quantum device to efficiently estimate an application-specific cost function or its gradient. By keeping the quantum circuit depth relatively short, VQAs mitigate hardware noise and may enable near-term applications including electronic structure~\cite{VQE,mcclean2016theory,bauer2016hybrid,jones2019variational}, dynamics~\cite{li2017efficient,cirstoiu2020variational,heya2019subspace,yuan2019theory}, optimization~\cite{qaoa2014,qaoaMaxCut2018,Crooks_2018,hadfield2019quantum}, linear systems~\cite{bravo-prieto2019,Xiaosi}, metrology~\cite{koczor2019variational, meyer2020variational}, factoring~\cite{Anschuetz2019Factoring}, compiling~\cite{QAQC,sharma2019noise,jones2018quantum}, and others~\cite{arrasmith2019variational,cerezo2019variational,cerezo2020variational,larose2019variational,verdon2019quantumHamiltonian,johnson2017qvector}.

The main open question for VQAs is their scalability to large problem sizes. While performing numerical heuristics for small or intermediate problem sizes is the norm for VQAs, deriving analytical scaling results is rare for this field. Noteworthy exceptions are some recent studies of the scaling of the gradient in VQAs with the number of qubits $n$~\cite{mcclean2018barren,holmes2021connecting,sharma2020trainability,cerezo2020cost,marrero2020entanglement,patti2020entanglement,volkoff2021large,cerezo2021higher,arrasmith2020effect}. For example, it was proven that the gradient vanishes exponentially in $n$ for randomly initialized, deep Hardware Efficient ansatzes~\cite{mcclean2018barren,holmes2021connecting} and dissipative quantum neural networks~\cite{sharma2020trainability}, and also for shallow depth with global cost functions~\cite{cerezo2020cost}. This vanishing gradient phenomenon is now referred to as barren plateaus in the training landscape.
Barren plateaus imply that in order to resolve gradients to a fixed precision, on average, an exponential number of shots need to be invested. This places an exponential resource burden on the training process of VQAs. Further, such effects are not avoided by adopting optimizers that use information about higher order derivatives \cite{cerezo2021higher} or gradient-free methods \cite{arrasmith2020effect}.  
Fortunately, investigations into barren plateaus have spawned several promising strategies to avoid them, including local cost functions~\cite{cerezo2020cost,uvarov2020barren}, parameter correlation~\cite{volkoff2021large}, pre-training~\cite{verdon2019learning}, and layer-by-layer training~\cite{grant2019initialization,skolik2020layerwise}. Such strategies give hope that perhaps VQAs may avoid the exponential scaling that otherwise would result from the exponential precision requirements of navigating through a barren plateau.

However, these works do not consider quantum hardware noise, and very little is known about the scalability of VQAs in the presence of noise. One of the main selling points of VQAs is noise mitigation, and indeed VQAs have shown evidence of optimal parameter resilience to noise in the sense that the global minimum of the landscape may be unaffected by noise~\cite{sharma2019noise,mcclean2016theory}. While some analysis has been done~\cite{xue2019effects, marshall2020characterizing,gentini2019noise}, an important open question, which has not yet been addressed, is how noise affects the asymptotic scaling of VQAs. More specifically, one can ask how noise affects the training process. If the effect of noise on trainability is not severe, and the optimal parameters can be found, then VQAs may be useful even in the presence of high decoherence in one of two ways. First, the end goal of certain algorithms such as the Quantum Approximate Optimization Algorithm (QAOA) \cite{farhi2016quantum} is to extract an optimized set of parameters, rather than the optimal cost value. 
Second, error mitigation can be used in conjunction with VQAs that display optimal parameter resilience. Intuitively, incoherent noise is expected to reduce the magnitude of the gradient and hence hinder trainability, and preliminary numerical evidence of this has been seen~\cite{kubler2020adaptive, arrasmith2020operator}, although the scaling of this effect has not been studied.

In this work, we analytically study the scaling of gradient for VQAs as a function of $n$, the circuit depth $L$, and a noise parameter $q<1$. We consider a general class of local noise models that includes depolarizing noise and certain kinds of Pauli noise. Furthermore, we investigate a general, abstract ansatz that allows us to encompass many of the important ansatzes in the literature, hence allowing us to make a general statement about VQAs. 
This includes the Quantum Alternating Operator Ansatz (QAOA) which is used for solving combinatorial optimization problems~\cite{qaoa2014,qaoaMaxCut2018,Crooks_2018,hadfield2019quantum} and the Unitary Coupled Cluster (UCC) Ansatz which is used in the Variational Quantum Eigensolver (VQE) to solve chemistry problems ~\cite{cao2019quantum,bartlett2007coupled,lee2018generalized}. This is also applicable for the Hardware Efficient Ansatz and the Hamiltonian Variational Ansatz (HVA) which are employed for various applications~\cite{kandala2017hardware,arute2020hartree,arute2020quantum,wecker2015progress,wiersema2020exploring}.
Our results also generalize to settings that allow for multiple input states or training data, as in machine learning applications, often called quantum neural networks~\cite{schuld2014quest,schuld2015introduction,biamonte2017quantum,beer2020training,abbas2020power}.

Our main result (Theorem~\ref{thm1}) is an upper bound on the magnitude of the gradient that decays exponentially with $L$, 
namely as $2^{-\kappa}$ with $\kappa = -L\log_2(q)$. Hence, we find that the gradient vanishes exponentially in the circuit depth. Moreover, it is typical to consider $L$ scaling as $\poly(n)$ (e.g., in the UCC Ansatz~\cite{lee2018generalized}), for which our main result implies an exponential decay of the gradient in $n$. We refer to this as a Noise-Induced Barren Plateau (NIBP). We remark that NIBPs can be viewed as concomitant to the cost landscape concentrating around the value of the cost for the maximally mixed state, 
and we make this precise in Lemma~\ref{lemma1}. See Fig.~\ref{fig:landscapes} for a schematic diagram of the NIBP phenomenon.

\begin{figure}[t]
    \centering
    \includegraphics[width=.9\columnwidth]{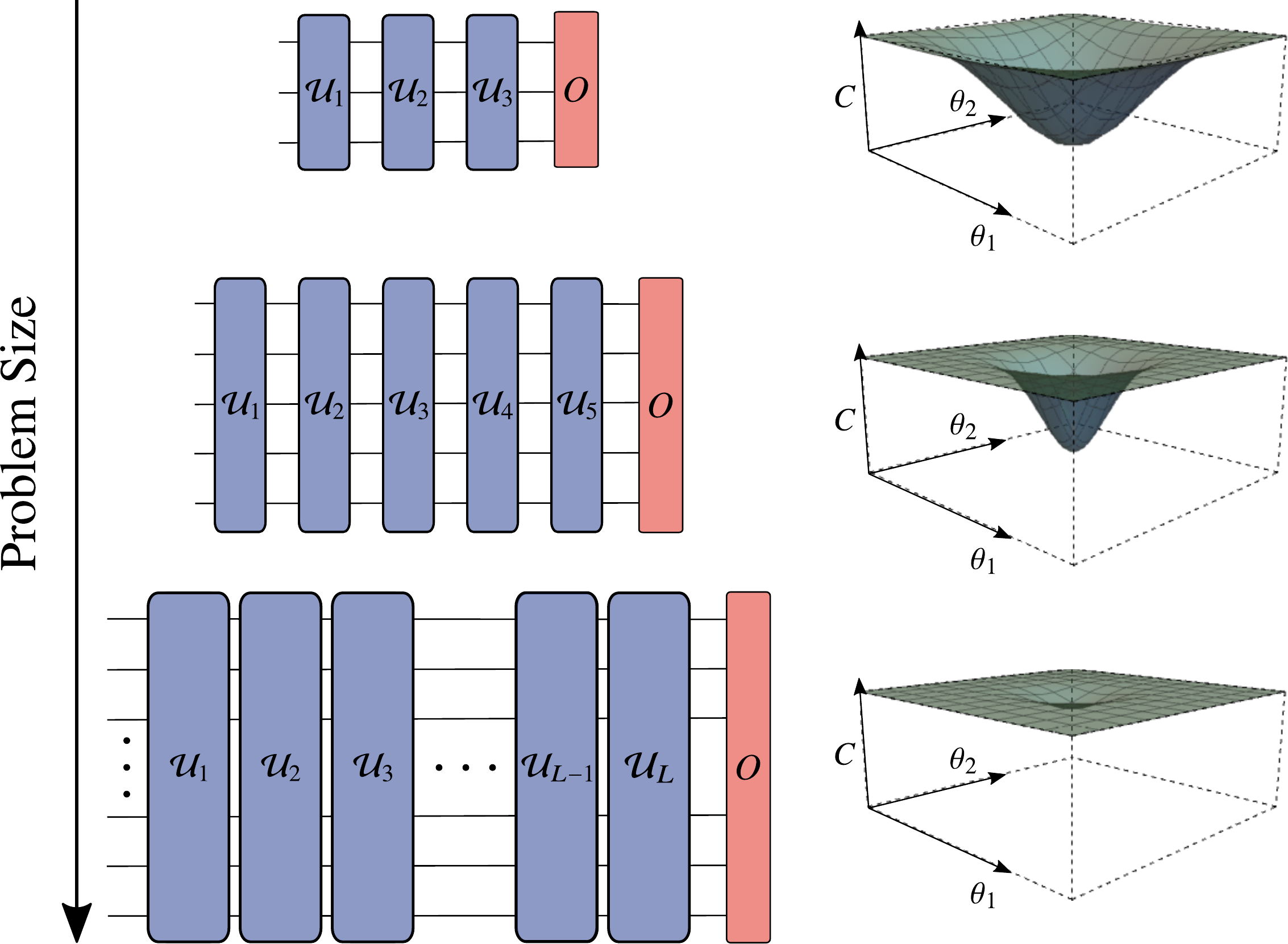}
    \caption{\textbf{Schematic diagram of the Noise-Induced Barren Plateau (NIBP) phenomenon.} For various applications such as chemistry and optimization, increasing the problem size often requires one to increase the depth $L$ of the variational ansatz. We show that, in the presence of local noise, the gradient vanishes exponentially in $L$ and hence exponentially in the number of qubits $n$ when $L$ scales linearly in $n$. This can be seen in the plots on the right, which show the cost function landscapes for a simple variational problem with local noise.  }
    \label{fig:landscapes}
\end{figure}

To be clear, any variational algorithm with a NIBP will have exponential scaling. In this sense, NIBPs destroy quantum speedup, as the standard goal of quantum algorithms is to avoid the typical exponential scaling of classical algorithms. NIBPs are conceptually distinct from the noise-free barren plateaus of Refs.~\cite{mcclean2018barren,sharma2020trainability,cerezo2020cost,marrero2020entanglement,patti2020entanglement,holmes2021connecting}. Indeed, strategies to avoid noise-free barren plateaus~\cite{cerezo2020cost,uvarov2020barren,volkoff2021large,verdon2019learning,grant2019initialization,skolik2020layerwise} do not appear to solve the NIBPs issue. 
 
The obvious strategy to address NIBPs is to reduce circuit complexity, or more precisely, to reduce the circuit depth. Hence, our work provides quantitative guidance for how small $L$ needs to be to potentially avoid NIBPs.

In what follows, we present our general framework followed by our main result. We also present two extensions of our main result, one involving correlated ansatz parameters and one allowing for measurement noise. The latter indicates that global cost functions exacerbate the NIBP issue. In addition, we provide numerical heuristics that illustrate our main result for MaxCut optimization with the QAOA, and an implementation of the HVA on superconducting hardware, both showing that NIBPs significantly impact this application.

\section{Results}

\subsection{General framework}\label{sec:gen-framework}

In this work we analyze a general class of parameterized ansatzes $U(\thv)$ that can be expressed as a product of $L$ unitaries sequentially applied by layers
\begin{equation}\label{eq:ansatz}
U(\thv) = U_L(\thv_L)\cdots U_2(\thv_2) \cdot U_1(\thv_1)\,.
\end{equation}
Here $\thv=\{\thv_l\}_{l=1}^L$ is a set of vectors of continuous parameters that are optimized to minimize a cost function $C$ that can be expressed as the expectation value of an operator $O$:  
\begin{equation}\label{eq:cost}
    C=\Tr[O U(\thv)\rho U\ad(\thv)]\,.
\end{equation}
As shown in Fig.~\ref{fig:setting}, $\rho$ is an $n$-qubit input state. Without loss of generality we assume that each $U_l(\thv_l)$ is given by 
\begin{equation}\label{eq:layerunitary}
    U_l(\thv_l)= \prod_{m} e^{-i \theta_{l m} H_{l m}} W_{l m}\,,
\end{equation}
where  $H_{l m}$ are Hermitian operators,  $\thv_l=\{\theta_{l m}\}$ are continuous parameters, and  $W_{l m}$ denote unparametrized gates. We expand  $H_{l m}$ and $O$ in the Pauli basis as
\begin{equation}\label{eq:HamiltonianDef}
    H_{l m} =\vec{\eta}_{l m} \cdot \vec{\sigma}_n= \sum_{i} \eta^{i}_{l m}\sigma_n^{i}\,,\quad O =\vec{\omega} \cdot \vec{\sigma}_n= \sum_{i} \omega^{i}\sigma_n^{i}\,,
\end{equation}
where $\sigma_n^{i}\in \{\id,X,Y,Z\}^{\otimes n}$   are Pauli strings, and   $\vec{\eta}_{l m}$ and $\vec{\omega}$ are real-valued vectors that specify the terms present in the expansion. Defining  $N_{l m}=\vert \vec{\eta}_{l m} \vert$ and $N_O=\vert \vec{\omega} \vert$ as the number of non-zero elements, i.e., the number of terms in the summations in Eq.~\eqref{eq:HamiltonianDef}, we say that $H_{l m}$ and $O$ admit  an   efficient Pauli decomposition if   $N_{l m},N_O\in \mathcal{O}( \poly (n))$, respectively.

We now briefly discuss how the  QAOA, UCC, and Hardware Efficient ansatzes fit into this general framework. We refer the reader to the Methods for additional details. In QAOA one sequentially alternates the action of two unitaries as    
\begin{equation}\label{eq:QAOA}
    U(\vec{\gamma},\vec{\beta})= e^{-i \beta_{p} H_M}  e^{-i \gamma_{p} H_P}\cdots e^{-i \beta_{1} H_M}  e^{-i \gamma_{1} H_P} \,,
\end{equation}
where $H_P$ and $H_M$ are the so-called problem and mixer Hamiltonian, respectively. We define $N_P\,(N_M)$ as the number of terms in the Pauli decomposition of $H_P\,(H_M)$.  
On the other hand,  Hardware Efficient ansatzes naturally fit into Eqs.~\eqref{eq:ansatz}--\eqref{eq:layerunitary} as they are usually composed of fixed gates (e.g, controlled NOTs),  and  parametrized gates (e.g., single qubit rotations). 
Finally, as detailed in the Methods, the UCC ansatz can be expressed as
\begin{equation}\label{eq:ucc}
        U(\vec{\theta}) =  \prod_{l m} U_{l m}(\theta_{l m})= \prod_{l m} e^{i  \theta_{l m} \sum_i   \mu^i_{l m}  \sigma^i_n}, 
\end{equation}
where $\mu^i_{l m} \in \{0,\pm1\}$, and where $\theta_{l m}$ are the coupled cluster amplitudes. 
Moreover, we denote $\widehat{N}_{l m} =|\vec{\mu}_{l m}|$ as the number of non-zero elements in  $\sum_i  \mu^i_{l m}  \sigma^i_n$.

As shown in Fig.~\ref{fig:setting}, we consider a noise model where local Pauli noise channels $\mathcal{N}_j$ act on each qubit $j$ before and after each unitary $U_l(\thv_l)$. The action of $\mathcal{N}_j$ on a local Pauli operator $\sigma\in\{X,Y,Z\}$  can be expressed as 
\begin{equation}\label{eq:noisemodel}
    \mathcal{N}_j(\sigma)=q_{\sigma}\sigma\,,
\end{equation}
where $-1< q_X,q_Y,q_Z<1$. Here, we characterize the noise strength with a single parameter $q=\max\{|q_X|,|q_Y|,|q_Z|\}$. Let $\mathcal{U}_l$ denote the channel that implements the unitary $U_l(\thv_l)$ and  let $\mathcal{N}=\mathcal{N}_1\otimes\cdots\otimes\mathcal{N}_n$ denote the $n$-qubit noise channel. Then, the noisy cost function is given by
\begin{equation} \label{eq:noisycost}
     \widetilde{C} = \Tr\left[ O\, \big(\mathcal{N}\circ \mathcal{U}_L\circ \cdots \circ \mathcal{N}\circ \mathcal{U}_1\circ \mathcal{N}\big) (\rho) \right]\,.
\end{equation}

\subsection{General analytical results}\label{sec:analytical-results}

\begin{figure}[t]
    \centering
    \includegraphics[width=.9\columnwidth]{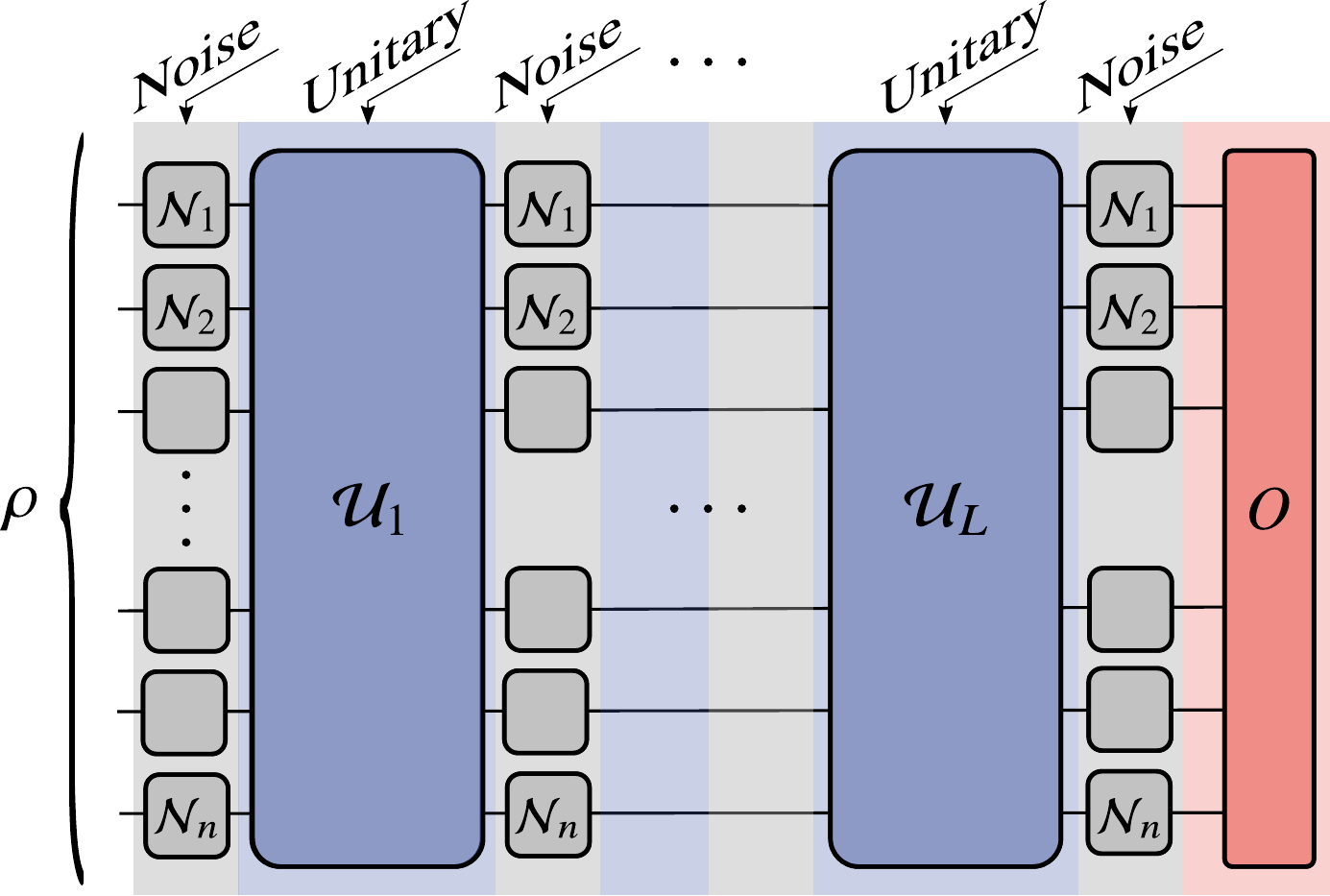}
    \caption{\textbf{Setting for our analysis.} An $n$-qubit input state $\rho$ is sent through a variational ansatz $U(\thv)$ composed of $L$ unitary layers $U_l(\thv_l)$ sequentially acting according to Eq.~\eqref{eq:ansatz}. Here,  $\mathcal{U}_l$ denotes the quantum channel that implements the unitary $U_l(\thv_l)$.   The parameters in the ansatz $\thv=\{\thv_l\}_{l=1}^L$ are trained to minimize a cost function that is expressed as the expectation value of an operator $O$ as in Eq.~\eqref{eq:cost}. We consider a noise model where local Pauli noise channels $\mathcal{N}_j$ act on each qubit $j$ before and after each unitary. }
    \label{fig:setting}
\end{figure}


There are some VQAs, such as the VQE~\cite{VQE} for chemistry and other physical systems, where it is important to accurately characterize the value of the cost function itself. We provide an important result below in Lemma~\ref{lemma1} that quantitatively bounds the cost function itself, and we envision that this bound will be especially useful in the context of VQE. On the other hand, there are other VQAs, such as those for optimization~\cite{qaoa2014,qaoaMaxCut2018,Crooks_2018,hadfield2019quantum}, compiling~\cite{QAQC,sharma2019noise,jones2018quantum}, and linear systems~\cite{bravo-prieto2019,Xiaosi}, where the key goal is to learn the optimal parameters and the precise value of the cost function is either not important or can be computed classically after learning the parameters. In this case, one is primarily concerned with trainability, and hence the gradient is a key quantity of interest. These applications motivate our main result in Theorem~\ref{thm1}, which bounds the magnitude of the gradient. We remark that trainability is of course also important for VQE, and hence Theorem~\ref{thm1} is also of interest for this application.

With this motivation in mind, we now present our main results. We first present our bound on the cost function, since one can view this as a phenomenon that naturally accompanies our main theorem. Namely, in the following lemma, we show that the noisy cost function concentrates around the corresponding value for the maximally mixed state.

\begin{lemma}[Concentration of the cost function] \label{lemma1}
  Consider an $L$-layered ansatz of the form in Eq.~\eqref{eq:ansatz}. Suppose that local Pauli noise of the form of Eq.~\eqref{eq:noisemodel} with noise strength $q$ acts before and after each layer as in Fig.~\ref{fig:setting}. Then, for a cost function $\widetilde{C}$ of the form in Eq.~\eqref{eq:noisycost}, the following bound holds
  \begin{equation}
      \Big\vert \widetilde{C} - \frac{1}{2^n}\Tr[O]\Big\vert\, \leq\, G(n) \,,
  \end{equation}
  where
  \begin{equation}
      G(n) = \sqrt{2\, \ln 2}\, N_O \|\vec{\omega}\|_\infty\,n^{1/2} q^{cL+1}\,.
  \end{equation}
  Here $\|\cdot\|_\infty$ is the infinity norm, $\|\cdot \|_1$ is the trace norm,
  $\vec{\omega}$ is defined in Eq.~\eqref{eq:HamiltonianDef}, and $N_O$ is the number of non-zero elements in the Pauli decomposition of $O$, and $c=1/(2\ln 2)$ is a constant.
\end{lemma}

This lemma implies the cost landscape exponentially concentrates on the value $\Tr [O]/2^n$ for large $n$, whenever the number of layers $L$ scales linearly with the number of qubits.
While this lemma has important applications on its own, particularly for VQE, it also provides intuition for the NIBP phenomenon, which we now state.

Let $\partial_{l m}\widetilde{C}=\partial \widetilde{C}/ \partial {\theta_{l m}}$ denote the partial derivative of the noisy cost function with respect to the $m$-th parameter that appears in the $l$-th layer  of the ansatz, as in Eq.~\eqref{eq:layerunitary}. For our main technical result, we upper bound $|\partial_{l m}\widetilde{C}|$ as a function of $L$ and $n$.

\begin{theorem}[Upper bound on the partial derivative] \label{thm1}
    Consider an $L$-layered ansatz as defined in Eq.~\eqref{eq:ansatz}. Let $\theta_{l m}$ denote the trainable parameter corresponding to the Hamiltonian $H_{l m}$ in the unitary $U_l(\vec{\theta}_l)$ appearing in the ansatz. Suppose that local Pauli noise of the form in Eq.~\eqref{eq:noisemodel} with noise parameter $q$ acts before and after each layer as in Fig.~\ref{fig:setting}. Then the following bound holds for the partial derivative of the noisy cost function
    \begin{equation}\label{eq:thm1}
       |\partial_{l m} \widetilde{C}|  \leq F(n)\,,
       \end{equation}
       where
    \begin{equation}\label{eq:bound-thm}
         F(n)= 
         \sqrt{8\ln{2}}\, N_O \|\vec{\omega}\|_\infty \big\|\vec{\eta}_{lm}\big\|_1  n^{1/2} q^{cL+1}\,, 
    \end{equation}
    and $\vec{\eta}_{lm}$ and  $\vec{\omega}$ are defined in Eq.~\eqref{eq:HamiltonianDef}, and $N_O$ is the number of non-zero Pauli terms in $O$, and $c=1/(2\ln 2)$ is a constant.
\end{theorem}

Let us now consider the asymptotic scaling of the function $F(n)$ in Eq.~\eqref{eq:bound-thm}. 
Under standard assumptions that $H_{l m}$ and $O$ in Eq.~\eqref{eq:HamiltonianDef} admit efficient Pauli decompositions, 
we now state that $F(n)$ decays exponentially in $n$, if $L$ grows linearly in $n$. 

\begin{corollary}[Noise-induced barren plateaus] \label{cor1}
     Let $N_{l m},N_O\in \mathcal{O}(\poly (n))$ and let ${\eta}^{i}_{l m},{\omega^j}\in \mathcal{O}(\poly (n))$ for all $i,j$. Then the upper bound $F(n)$ in Eq.~\eqref{eq:bound-thm} vanishes exponentially in $n$ as
    \begin{equation}\label{eqnFnScaling} 
        F(n) \in \mathcal{O}(2^{-\alpha n})\,,
    \end{equation}
    for some positive constant $\alpha$ if we have
    \begin{equation} \label{eq:Lcondition}
        L \in \Omega(n)\,.
    \end{equation}
\end{corollary}

The asymptotic scaling in Eq.~\eqref{eqnFnScaling} is independent of $l$ and $m$, i.e., the scaling is blind to the layer, or the parameter within the layer, for which the derivative is taken. 
This corollary implies that when  Eq.~\eqref{eq:Lcondition} holds, i.e. $L$ grows at least linearly in $n$, the partial derivative $|\partial_{l m} \widetilde{C}|$ exponentially vanishes in $n$ across the entire cost landscape. In other words, one observes a Noise-Induced Barren Plateau (NIBP). We note that Eq.~\eqref{eq:Lcondition} is satisfied for all $q<1$. That is, NIBPs occur regardless of the noise strength, it only changes the severity of the exponential scaling.

In addition, Corollary~\ref{cor1} implies that NIBPs are conceptually different from noise-free barren plateaus. First, NIBPs are independent of the parameter initialization strategy or the locality of the cost function. Second, NIBPs exhibit exponential decay of the gradient itself; not just of the variance of the gradient, which is the hallmark of noise-free barren plateaus. Noise-free barren plateaus allow the global minimum to sit inside deep, narrow valley in the landscape~\cite{cerezo2020cost}, whereas NIBPs flatten the entire landscape.

One of the strategies to avoid the noise-free barren plateaus is to correlate parameters, i.e., to make a subset of the parameters equal to each other~\cite{volkoff2021large}. We generalize Theorem \ref{thm1} in the following remark to accommodate such a setting, consequently showing that such correlated or degenerate parameters do not help in avoiding NIBPs. In this setting, the result we obtain in Eq.~\eqref{eqnDegenerateResult} below is essentially identical to that in Eq.~\eqref{eq:bound-thm} except with an additional factor quantifying the amount of degeneracy.

\begin{remark}[Degenerate parameters]\label{remark1}
   Consider the ansatz defined in Eqs.~\eqref{eq:ansatz} and \eqref{eq:layerunitary}. Suppose there is a subset $G_{st}$ of the set $\{\theta_{l m}\}$ in this ansatz such that $G_{st}$ consists of $g$ parameters that are degenerate: 
\begin{equation}
    G_{st} = \big\{  \theta_{l m} \;|\; \theta_{l m}=\theta_{st} \big\}\,.
\end{equation}    
Here, $\theta_{st}$ denotes the parameter in $G_{st}$ for which $N_{l m}  \|\vec{\eta}_{l m} \|_\infty$ takes the largest value in the set. ($\theta_{st}$ can also be thought of as a reference parameter to which all other parameters are set equal in value.) Then the partial derivative of the noisy cost with respect to $\theta_{st}$ is bounded as
\begin{equation}\label{eqnDegenerateResult}
    |\partial_{st} \widetilde{C}| \leq \sqrt{8\ln{2}}\,g\, N_O \|\vec{\omega}\|_\infty \|\vec{\eta}_{lm}\|_\infty  n^{1/2} q^{cL+1}, 
\end{equation}
at all points in the cost landscape.
\end{remark}

Remark~\ref{remark1} is especially important in the context of the QAOA and the UCC ansatz, as discussed below. We note that, in the general case, a unitary of the form of Eq.~\eqref{eq:layerunitary} cannot be implemented as a single gate on a physical device. In practice one needs to compile the unitary into a sequence of native gates. Moreover, Hamiltonians with non-commuting terms are usually approximated with techniques such as Trotterization. This compilation overhead potentially leads to a sequence of gates that grows with $n$. Remark \ref{remark1} enables us to account for such scenarios, and we elaborate on its relevance to specific applications in the next subsection.

In reality, noise on quantum hardware can be non-local. For instance in certain systems one can have cross-talk errors or coherent errors. We address such extensions to our noise model in the following remark.

\begin{remark}[Extensions to the noise model]\label{remark:extension_noisemodel}
    Consider a modification to each layer of noise $\mathcal{N}$ in Eq.~\eqref{eq:noisycost}   to include additional $k$-local Pauli noise and correlated coherent (unitary) noise across multiple qubits. Under such extensions to the noise model, we obtain the same scaling results as those obtained in Lemma \ref{lemma1} and Theorem~\ref{thm1}. We discuss this in more detail in Supplementary Note \hyperref[sec:remark2-supp]{5}.
\end{remark}

Finally, we present an extension of our main result to the case of measurement noise. Consider a model of measurement noise where each local measurement independently has some bit-flip probability given by $(1-q_M)/2$, which we assume to be symmetric with respect to the 0 and 1 outcomes. This leads to an additional reduction of our bounds on the cost function and its gradient that depends on the locality of the observable $O$.

\begin{proposition}[Measurement noise] \label{prop1}
    Consider expanding the observable $O$ as a sum of Pauli strings, as in Eq.~\eqref{eq:HamiltonianDef}. Let $w$ denote the minimum weight of these strings, where the weight is defined as the number of non-identity elements for a given string. In addition to the noise process considered in Fig.~\ref{fig:setting}, suppose there is also measurement noise consisting of a tensor product of local bit-flip channels with bit-flip probability $(1-q_M)/2$. Then we have
    \begin{equation}\label{eq:prop-b1}
        \left\vert\widetilde{C} - \frac{1}{2^n}\Tr[O]\right\vert \leq q_M^{w}\, G(n) \,,
    \end{equation}
    and
    \begin{equation}
        |\partial_{l m} \widetilde{C}| \leq q_M^{w} F(n) \,,
    \end{equation}
    where $G(n)$ and $F(n)$ are defined in Lemma~\ref{lemma1} and Theorem~\ref{thm1}, respectively.
\end{proposition}

Proposition~\ref{prop1} goes beyond the noise model considered in Theorem~\ref{thm1}. It shows that in the presence of measurement noise there is an additional contribution from the locality of the measurement operator. It is interesting to draw a parallel between Proposition~\ref{prop1} and noise-free barren plateaus, which have been shown to be cost-function dependent and in particular depend on the locality of the observable $O$~\cite{cerezo2020cost}. The bounds in Proposition~\ref{prop1} similarly depend on the locality of $O$. For example, when  $w=n$, i.e., global observables, the factor $q_M^w$ will hasten the exponential decay. On the other hand, when $w=1$, i.e., local observables, the scaling is unaltered by measurement noise. In this sense, a global observable exacerbates the NIBP issue by making the decay more rapid with $n$.

\subsection{Application-specific analytical results}

Here we investigate the implications of our results from the previous subsection for two applications: optimization and chemistry. In particular, we derive explicit conditions for NIBPs for these applications. 
These conditions are derived in the setting where Trotterization is used, but other compilation strategies incur similar asymptotic behavior.
We begin with the QAOA for optimization and then discuss the UCC ansatz for chemistry. Finally, we make a remark about the Hamiltonian Variational Ansatz (HVA),
as well as remark that our results also apply to a generalized cost function that can employ training data. 

\begin{corollary}[Example: QAOA]
\label{cor:qaoa}
Consider the QAOA with $2p$ trainable parameters, as defined in Eq.~\eqref{eq:QAOA}. Suppose that the implementation of unitaries corresponding to the problem Hamiltonian $H_P$ and the mixer Hamiltonian $H_M$ require $k_P$- and $k_M$-depth circuits, respectively.
If local Pauli noise of the form in Eq.~\eqref{eq:noisemodel} with noise parameter $q$ acts before and after each layer of native gates, then we have
\small
    \begin{align}
        |\partial_{\beta_l} \widetilde{C}| \, &\leq \sqrt{8\ln{2}}\, {g_{l,P}}N_P \|\vec{\omega}\|_\infty \big\|\vec{\eta}_{P} \big\|_1  n^{1/2} q^{c(k_P+k_M)p+1}, \label{eq:qaoabound1}\\
        |\partial_{\gamma_l} \widetilde{C}| \, &\leq  \sqrt{8\ln{2}}\, {g_{l,M}}N_P   \|\vec{\omega} \|_\infty \big\|\vec{\eta}_{P} \big\|_1 \,n^{1/2}  q^{c(k_P+k_M)p+1}, \label{eq:qaoabound2}
    \end{align}
\normalsize
    for any choice of parameters $\beta_l,\gamma_l$, and where $O=H_P$ in Eq.~\eqref{eq:cost}. Here $g_{l,P}$ and $g_{l,M}$ are the respective number of native gates parameterized by $\beta_l$ and $\gamma_l$ according to the compilation.
\end{corollary}

Corollary \ref{cor:qaoa} follows from Remark \ref{remark1} and it has interesting implications for the trainability of the QAOA. From Eqs.~\eqref{eq:qaoabound1} and \eqref{eq:qaoabound2}, NIBPs are guaranteed if $pk_P$ scales linearly in $n$. This can manifest itself in a number of ways, which we explain below. 

First, we look at the depth $k_P$ required to implement one application of the problem unitary. Graph problems containing vertices of extensive degree such as the Sherrington-Kirkpatrick model inherently require $\Omega(n)$ depth circuits to implement \cite{arute2020quantum}. On the other hand, generic problems mapped to hardware topologies also have the potential to incur $\Omega(n)$ depth or greater in compilation cost. For instance, implementation of MaxCut and $k$-SAT using SWAP networks on circuits with 1-D connectivity requires depth $\Omega(n)$ and $\Omega(n^{k-1})$ respectively \cite{Crooks_2018, o2019generalized}. Such mappings with the aforementioned compilation overhead for $k\geq2$ are guaranteed to encounter NIBPs even for a fixed number of rounds $p$.

Second, it appears that $p$ values that grow at least lightly with $n$ may be needed for quantum advantage in certain optimization problems (for example, \cite{Bravyi2019ObstaclesTS, Wang2018fermionic, hastings2019classical,jiang2017near}). In addition, there are problems employing the QAOA that explicitly require $p$ scaling as $\poly(n)$ \cite{Akshay2020Reachability, Anschuetz2019Factoring}. Thus, without even considering the compilation overhead for the problem unitary, these QAOA problems may run into NIBPs particularly when aiming for quantum advantage. Moreover, weak growth of $p$ with $n$ combined with compilation overhead could still result in an NIBP.

Finally, we note that above we have assumed the contribution of $k_P$ dominates that of $k_M$. However, it is possible that for choice of more exotic mixers \cite{hadfield2019quantum}, $k_M$ also needs to be carefully considered to avoid NIBPs.

\begin{corollary}[Example: UCC] \label{cor:ucc}
Let $H$ denote a molecular Hamiltonian of a system of $M_e$ electrons. Consider the UCC ansatz as defined in Eq.~\eqref{eq:ucc}.  If local Pauli noise of the form in Eq.~\eqref{eq:noisemodel} with noise parameter $q$ acts before and after every $U_{l m}(\theta_{l m})$ in Eq.~\eqref{eq:ucc}, then we have
\begin{align}\label{eq:ucc-cor}
 |\partial_{\theta_{l m}} \widetilde{C}| \leq \sqrt{8\ln{2}}\,\widehat{N}_{l m} N_H \Vert \vec{\omega}\Vert_{\infty} \,n^{1/2} q^{cL+1}, 
\end{align}
for any coupled cluster amplitude $\theta_{l m}$,  and where $O=H$ in Eq.~\eqref{eq:cost}. 
\end{corollary}

Corollary \ref{cor:ucc} allows us to make general statements about the trainability of UCC ansatz. We present the details for the standard UCC ansatz with single and double excitations from occupied to virtual orbitals \cite{cao2019quantum,mcardle2020quantum} (see Methods for more details). Let $M_o$ denote the total number of spin orbitals. Then at least $n=M_o$ qubits are required to simulate such a system and the number of variational parameters grows as $\Omega(n^2 M_e^2)$ \cite{romero2018strategies,o2019generalized}. To implement the UCC ansatz on a quantum computer, the excitation operators are first mapped to Pauli operators using Jordan-Wigner or Bravyi-Kitaev mappings \cite{ortiz2001quantum,bravyi2002fermionic}. Then, using first-order Trotterization and employing SWAP networks \cite{o2019generalized}, the UCC ansatz can be implemented in $\Omega(n^2 M_e)$ depth, while assuming 1-D connectivity of qubits \cite{o2019generalized}. Hence for the UCC ansatz, approximated by single- and double-excitation operators, the upper bound in Eq.~\eqref{eq:ucc-cor} (asymptotically) vanishes exponentially in $n$. 

To target strongly correlated states for molecular Hamiltonians, one can employ a UCC ansatz that includes additional, generalized excitations \cite{nooijen2000can,wecker2015progress}. A $\Omega(n^3)$ depth circuit is required to implement the first-order Trotterized form of this ansatz \cite{o2019generalized}. Hence NIBPs become more prominent for generalized UCC ansatzes. Finally, we remark that a sparse version of the UCC ansatz can be implemented in $\Omega(n)$ depth \cite{o2019generalized}. NIBPs still would occur for such ansatzes.

Additionally, we can make the following remark about the Hamiltonian Variational Ansatz (HVA). As argued in \cite{wecker2015progress,ho2019efficient, cade2019strategies}, the HVA has the potential to be an effective ansatz for quantum many-body problems.
\begin{remark}[Example: HVA]
The HVA can be thought of as a generalization of the QAOA to more than two non-commuting Hamiltonians. It is remarked in Ref.~\cite{wiersema2020exploring} that for problems of interest the number of rounds $p$ scales linearly in $n$. Thus, considering this growth of $p$ and also the potential growth of the compiled unitaries with $n$, the HVA has the potential to encounter NIBPs, by the same arguments made above for the QAOA (e.g., Corollary~\ref{cor:qaoa}).
\end{remark}

\begin{remark}[Quantum Machine Learning] \label{remark:QML}
Our results can be extended to generalized cost functions of the form $C_{\textrm{train}}= \sum_i \Tr[O_i U(\thv)\rho_i U\ad(\thv)]$ where $\{O_i\}$ is a set of operators each of the form \eqref{eq:HamiltonianDef} and $\{\rho_i\}$ is a set of states. This can encapsulate certain quantum machine learning settings~\cite{schuld2014quest,schuld2015introduction,biamonte2017quantum,beer2020training,abbas2020power} that employ training data $\{ \rho_i\}$. As an example of an instance where NIBPs can occur, in one study \cite{abbas2020power} an ansatz model has been proposed that requires at least linear circuit depth in $n$.
\end{remark}


\subsection{Numerical simulations of the QAOA}

To illustrate the NIBP phenomenon beyond the conditions assumed in our analytical results, we numerically  implement the QAOA to solve MaxCut combinatorial optimization problems. We employ a realistic noise model obtained from gate-set tomography on the IBM Ourense superconducting qubit device. In the Methods we provide additional details on the noise model and the optimization method employed.

Let us first  recall that a MaxCut problem is specified by a graph $G=(V,E)$ of nodes $V$ and edges $E$. The goal is to partition the nodes of $G$ into two sets which maximize the number of edges connecting nodes between sets. Here, the QAOA problem Hamiltonian is given by
\begin{equation}
    H_P=-\frac{1}{2}\sum_{ij \in E} C_{ij}(\id-Z_i Z_j)\,,
\end{equation}
where $Z_i$ are local Pauli operators on qubit (node) $i$, $C_{ij}=1$ if the nodes are connected and $C_{ij}=0$ otherwise.

We analyze performance in two settings. First, we fix the problem size at $n=5$ nodes (qubits) and vary the number of rounds $p$ (Fig.~\ref{fig:heuristics}). Second, we fix the number of rounds of QAOA at $p=4$ and vary the problem size by increasing the number of nodes (Fig.~\ref{fig:heuristics2}).


In order to quantify performance for a given $n$ and~$p$, we randomly generate $100$ graphs according to the Erd\H os-R\'enyi model~\cite{erdds1959random}, such that each graph $G$ is chosen uniformly at random from the set of all graphs of $n$ nodes.
For each graph we run $10$ instances of the parameter optimization, and we select the run that achieves the smallest energy. At each optimization step the cost is estimated with $1000$ shots. Performance is quantified by the average approximation ratio when training the QAOA in the presence and absence of noise. The  approximation ratio is defined as the lowest energy obtained via optimizing divided by the exact ground state energy of $H_P$. 

\begin{figure}[t]
    \centering
    \includegraphics[width=\columnwidth]{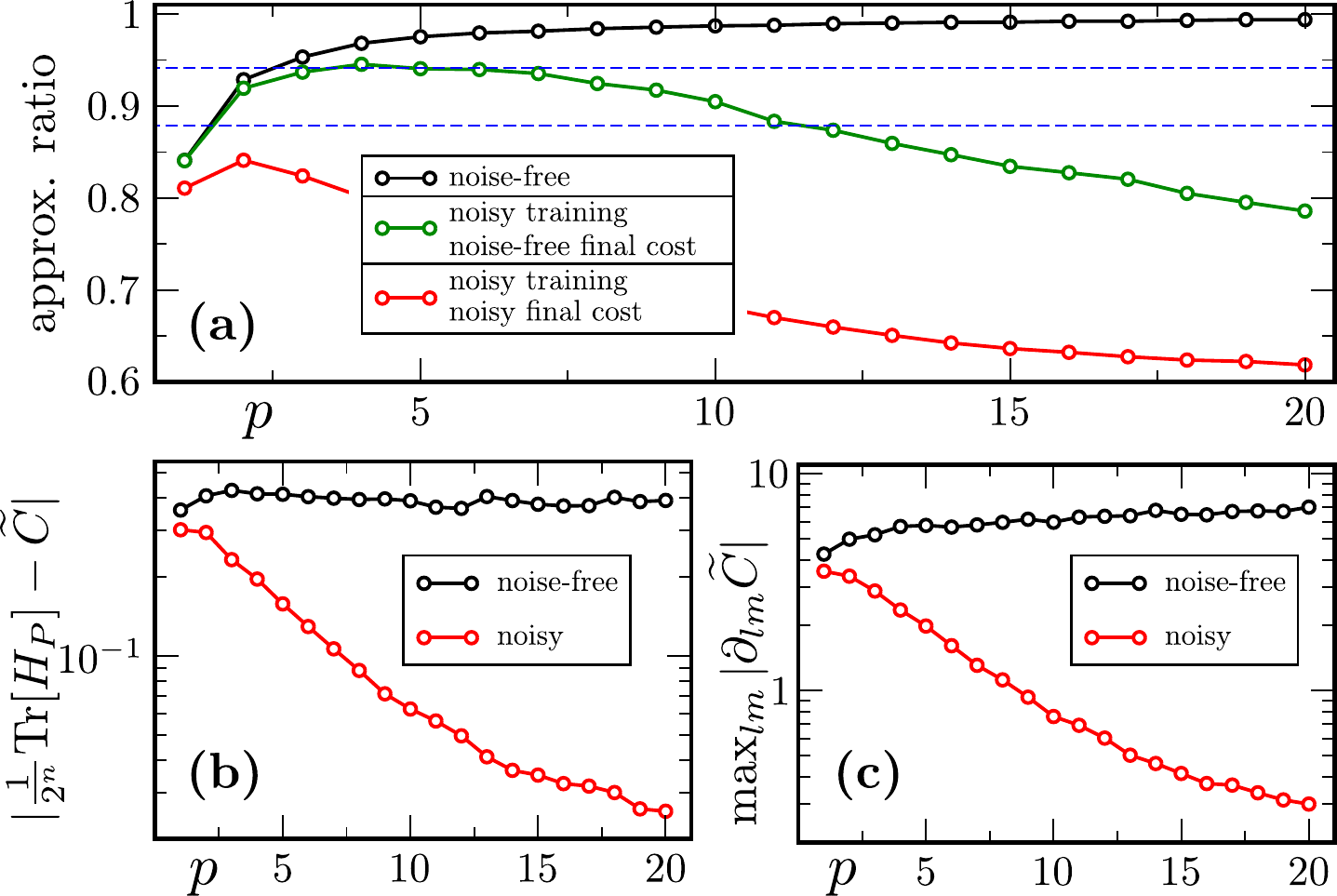} \caption{\textbf{QAOA heuristics in the presence of realistic hardware noise: increasing number of rounds for fixed problem size}. (a)~The approximation ratio averaged over 100 random graphs of 5 nodes is plotted versus number of rounds $p$. The black, green, and red curves respectively correspond to noise-free training, noisy training with noise-free final cost evaluation, and noisy  training with noisy final cost evaluation.  
    The performance of noise-free training increases with $p$, similar to the results in Ref.~\cite{Crooks_2018}. The green curve shows that the training process itself is hindered by noise, with the performance decreasing steadily with $p$ for $p>4$. The dotted blue lines correspond to known lower and upper bounds on classical performance in polynomial time: respectively the performance guarantee of the Goemans-Williamson algorithm \cite{goemans1995improved} and the boundary of known NP-hardness \cite{arora1998proof,haastad2001some}.
    (b) The deviation of the cost from $\Tr[H_P]/2^n$ (averaged over graphs and parameter values) is plotted versus $p$. As $p$ increases, this deviation decays approximately exponentially with $p$ (linear on the log scale). (c) The absolute value of the largest partial derivative, averaged over graphs and parameter values, is plotted versus $p$. The partial derivatives decay approximately exponentially with $p$, showing evidence of Noise-Induced Barren Plateaus (NIBPs).}
    \label{fig:heuristics}
\end{figure}

\begin{figure}[ht]
    \centering
    \includegraphics[width=\columnwidth]{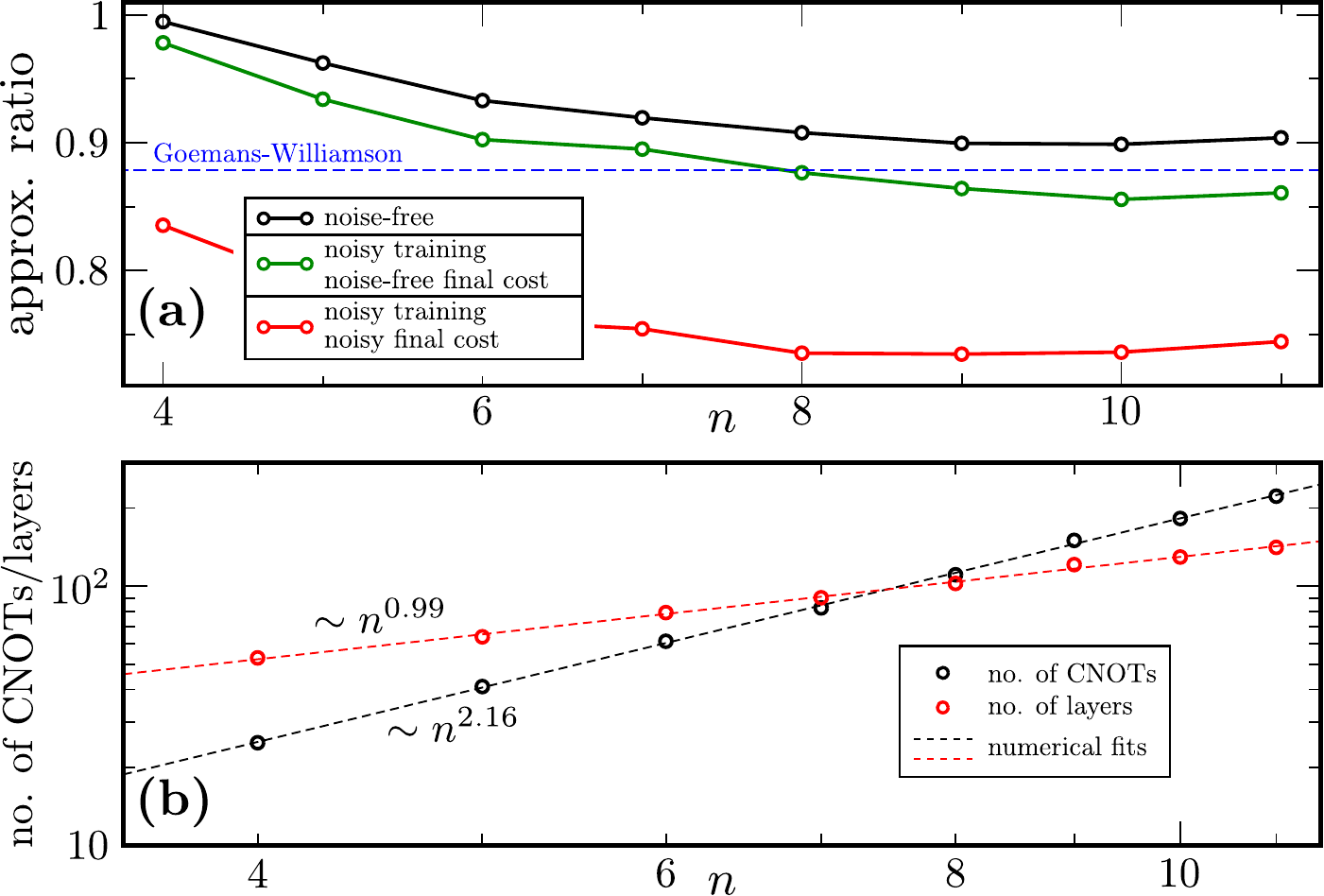} \caption{\textbf{QAOA heuristics in the presence of realistic hardware noise: increasing problem size for a fixed number of rounds}. The approximation ratio averaged over 60 random graphs of increasing number of nodes $n$ and fixed number of rounds $p=4$ is plotted. The black, green, and red curves respectively correspond to noise-free training, noisy training with noise-free final cost evaluation, and noisy training with noisy final cost evaluation. (a) For a problem size of 8 nodes or larger, the noisily-trained approximation ratio falls below the performance guarantee of the classical Goemans-Williamson algorithm. (b) The depth of the circuit (red curve) scales linearly with the number of qubits, confirming we are in a regime where we would expect to observe Noise-Induced Barren Plateaus.}
    \label{fig:heuristics2}
\end{figure}

In our first setting we observe in Fig.~\ref{fig:heuristics}(a) that when training in the absence of noise, the approximation ratio increases with $p$. However, when training in the presence of noise the performance decreases for $p>2$. This result is in accordance with Lemma~\ref{lemma1}, as the cost function value concentrates around $\Tr[H_P]/2^n$ as $p$ increases. This concentration phenomenon can also be seen clearly in Fig.~\ref{fig:heuristics}(b), where in fact we see evidence of exponential decay of cost value with $p$.  


In addition, we can see the effect of NIBPs as Fig.~\ref{fig:heuristics}(a) also depicts the value of the approximation ratio computed without noise by utilizing the parameters obtained via noisy training. Note that evaluating the cost in a noise-free setting has practical meaning, since the classicality of the Hamiltonian allows one to compute the cost on a (noise-free) classical computer, after training the parameters. For $p>4$ this approximation ratio decreases, meaning that as $p$ becomes larger it becomes increasingly hard to find a minimizing direction to navigate through the cost function landscape. Moreover, the effect of NIPBs is evident in  Fig.~\ref{fig:heuristics}(c) where we depict the average absolute value  of the largest cost function partial derivative (i.e.,  $\max_{l m}|\partial_{l m}\widetilde{C}|$). This plot shows an exponential decay of the partial derivative with $p$ in accordance with Theorem~\ref{thm1}. 

Finally, in Fig.~\ref{fig:heuristics}(a) we contextualize our results with previously known two-sided bounds on classical polynomial-time performance. The lower bound corresponds to the performance guarantee of the classical Goemans-Williamson algorithm \cite{goemans1995improved}, whilst the upper bound is at the value $16/17$ which is the approximation ratio beyond which Max-Cut is known to be NP-hard \cite{arora1998proof,haastad2001some}.  

In our second setting we find complementary results. In Fig.~\ref{fig:heuristics2}(a) we observe that at a problem size of $8$ qubits or larger, 4 rounds of QAOA trained on the noisy circuit falls short of the performance guarantees of the classical Goemans-Williamson algorithm. As we increase the number of qubits, we also observe this increases the depth of the circuit linearly (Fig.~\ref{fig:heuristics2}(b)), thus confirming we are in a regime of NIBPs.

Our numerical results show that training the QAOA in the presence of a realistic noise model significantly affects the performance. The concentration of cost and the NIBP phenomenon are both also clearly visible in our data. Even though we observe performance for $n=5$ and $p=4$ that is NP-hard to achieve classically, any possible advantage would be lost for large problem sizes or circuit depth due to bad scaling.
Hence, noise appears to be a crucial factor to account for when attempting to understand the performance of the QAOA.

\subsection{Implementation of the HVA on superconducting hardware}

We further demonstrate the NIBP phenomenon in a realistic hardware setting by implementing the Hamiltonian Variational Ansatz (HVA) on the IBM Quantum \textit{ibmq\_montreal} $27$-qubit superconducting device. At time of writing this holds the record for the largest quantum volume measured on an IBM Quantum device, which was demonstrated on a line of 6 qubits \cite{jurcevic2021demonstration}. 

We implement the HVA for the Transverse Field Ising Model as considered in Ref.~\cite{wiersema2020exploring}, with a local measurement $O=Z_0 Z_1$ on the first two qubits of the Ising chain. We assign the number of layers $L$ of the ansatz to increase linearly with the number of qubits $n$ according to the relationship $L=n-1$. In order to minimize SWAP gates used in transpilation (and the accompanying noise that they incur), we modify each layer of the HVA ansatz to only include entangling gates between locally connected qubits. 

Figure~\ref{fig:hardware} plots the partial derivative of the cost function with respect to the parameter in the final layer of the ansatz, averaged over 100 random parameter sets. We also plot averaged cost differences from the corresponding maximally mixed values, as well as the variance of both quantities.  In the noise-free case both the partial derivative and cost value differences decrease at a sub-exponential rate. Meanwhile, in the noisy case we observe that both the partial derivatives and cost value differences vanish exponentially until their variances reach the same order of magnitude as the shot noise floor. (As the shot budget on the IBM Quantum device is limited, this leads to a background of shot noise, and we plot the order of magnitude of this with a dotted line.) This explicitly demonstrates that the problem of barren plateaus is one of \emph{resolvability}. In principle, if one has access to exact cost values and gradients one may be able to navigate the cost landscape, however, the number of shots required to reach the necessary resolution increases exponentially with $n$.

\begin{figure}[t]
    \centering
    \includegraphics[width=\columnwidth]{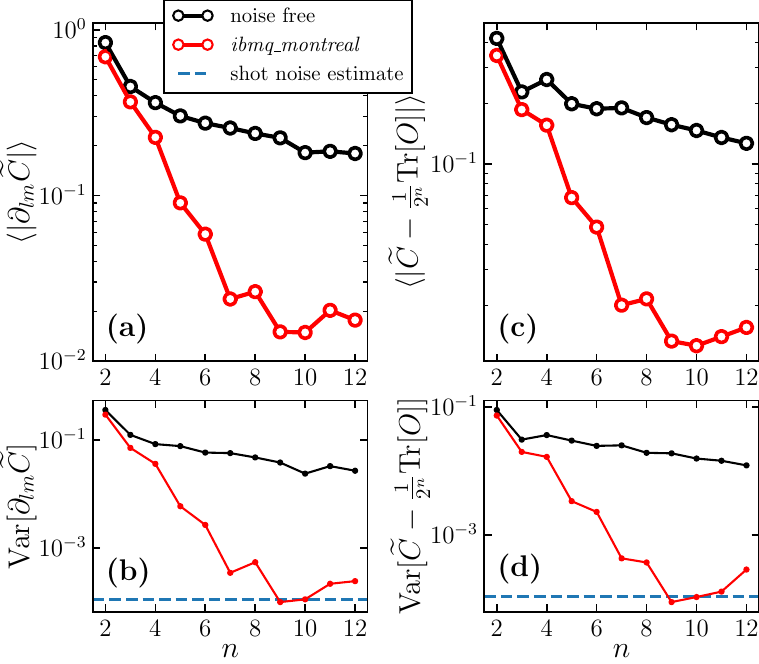} \caption{ \textbf{Implementation on the \emph{ibmq\_montreal} superconducting-qubit device}. We consider the HVA with the number of layers growing linearly in the number of qubits, $n$. a) The average magnitude of the partial derivative of the noisy and noise-free cost, with respect to the parameter in the final layer, is plotted versus $n$. The average is taken over 100 randomly selected parameter sets. As $n$ increases, the noisy average partial derivative decreases approximately exponentially, until around $n=9$. This shows evidence of Noise Induced Barren Plateaus on real quantum hardware. b) The deviation from exponential scaling can be understood by observing that it coincides with the point that the variance of the noisy partial derivatives reaches the same order of magnitude as the shot noise given by a finite sample budget of 8192 shots. Thus, from this point onward we expect fluctuations in the partial derivative to be dominated by shot noise, and gradients to be unresolvable.  c) The difference of the cost value from its corresponding maximally mixed value is plotted versus $n$. d) The variance of this difference is plotted versus $n$. Both these quantities also show exponential decay until the variance of cost difference approaches the shot noise floor, which shows evidence of exponential cost concentration on this device.}
    \label{fig:hardware}
\end{figure}

\section{Discussion}

The success of NISQ computing largely depends on the scalability of Variational Quantum Algorithms (VQAs), which are widely viewed as the best hope for near-term quantum advantage for various applications. Only a small handful of works have analytically studied VQA scalability, and there is even less known about the impact of noise on their scaling. Our work represents a breakthrough in understanding the effect of local noise on VQA scalability. We rigorously prove two important and closely related phenomena: the exponential concentration of the cost function in Lemma~\ref{lemma1} and the exponential vanishing of the gradient in Theorem~\ref{thm1}. We refer to the latter as a Noise-Induced Barren Plateau (NIBP). Like noise-free barren plateaus, NIBPs require the precision and hence the algorithmic complexity to scale exponentially with the problem size. Thus, avoiding NIBPs is necessary for a VQA to have any hope of exponential quantum speedup. 

NIBPs have conceptual differences from noise-free barren plateaus~\cite{mcclean2018barren,sharma2020trainability,cerezo2020cost,marrero2020entanglement,patti2020entanglement,holmes2021connecting} as the gradient vanishes with increasing problem size at every point on the cost function landscape, rather than probabilistically. As a consequence, NIBPs cannot be addressed by layer-wise training, correlating parameters and other strategies~\cite{cerezo2020cost,uvarov2020barren,volkoff2021large,verdon2019learning,grant2019initialization,skolik2020layerwise}, all of which can help avoid noise-free barren plateaus. We explicitly demonstrate this in Remark~\ref{remark1} for the parameter correlation strategy. Similar to noise-free barren plateaus, NIBPs present a problem for trainability even when utilizing gradient-free optimizers \cite{arrasmith2020effect} (e.g.~simplex-based methods such as \cite{nelder1965simplex} or methods designed specifically for quantum landscapes  \cite{koczor2020quantum})  or optimization strategies  that use higher-order derivatives \cite{cerezo2021higher}. At the moment, the only strategies we are aware of for avoiding NIBPs are: (1) reducing the hardware noise level, or (2) improving the design of variational ansatzes such that their circuit depth scales more weakly with $n$. Our work provides quantitative guidance for how to develop these strategies.

We emphasize that na\"ive mitigation strategies such as artificially increasing gradients cannot remove the exponential scaling of NIBPs as this simply increases the variance of any finite-shot evaluation of derivatives, and it does not improve the resolvability of the landscape. This argument extends simply to include any error mitigation strategy that implements an affine map to cost values \cite{czarnik2020error,montanaro2021error,vovrosh2021efficient, rosenberg2021experimental, he2020zero, shaw2021classical, arute2020observation}. Further, most error mitigation techniques consist only of postprocessing noisy circuits. Thus, we deem it unlikely many strategies can remove exponential NIBP scaling as information about the cost landscape has fundamentally been lost (or at least been made exponentially inaccessible). This is in contrast to error correction where information is protected and recovered. However, in general it is an open question as to whether or not error mitigation strategies can mitigate NIBPs, and we leave this question for future work.

An elegant feature of our work is its generality, as our results apply to a wide range of VQAs and ansatzes. This includes the two most popular ansatzes, QAOA for optimization and UCC for chemistry, which Corollaries~\ref{cor:qaoa} and~\ref{cor:ucc} treat respectively. In recent times QAQA, UCC, and other physically motivated ansatzes have be touted as the potential solution to trainability issues due to (noise-free) barren plateaus, while Hardware Efficient ansatzes, which minimize circuit depth, have been regarded as problematic. Our work swings the pendulum in the other direction: any additional circuit depth that an ansatz incorporates (regardless of whether it is physically motivated) will hurt trainability and potentially lead to a NIBP. This suggests that Hardware Efficient ansatzes are in fact worth exploring further, provided one has an appropriate strategy to avoid noise-free barren plateaus. This claim is supported by recent state-of-the-art implementations for optimization~\cite{arute2020quantum} and chemistry~\cite{arute2020hartree} using such ansatzes. Our work also provides additional motivation towards the pursuit of adaptive ansatzes \cite{bilkis2021semi, grimsley2019adaptive, tang2021qubit, zhang2021mutual, rattew2019domain, chivilikhin2020mog, cincio2021machine, cincio2018learning, du2020quantum} that reduce circuit depth.

We believe our work has particular relevance to optimization. For combinatorial optimization problems, such as MaxCut on 3-regular graphs, the compilation of a single instance of the problem unitary $e^{-i \gamma H_P}$ can require an $\Omega(n)$-depth circuit~\cite{arute2020quantum}. Therefore, for a constant number of rounds $p$ of the QAOA, the circuit depth grows at least linearly with $n$. From Theorem \ref{thm1}, it follows that NIBPs can occur for practical QAOA problems, even for constant number of rounds.  Furthermore, even neglecting the aforementioned linear compilation overhead, NIBPs are guaranteed (asymptotically) if $p$ grows in $n$. Such growth has been shown to be necessary in certain instances of MaxCut~\cite{Bravyi2019ObstaclesTS} as well as for other optimization problems~\cite{Akshay2020Reachability, Anschuetz2019Factoring}, and hence NIBPs are especially relevant in these cases.

While it is well known that decoherence ultimately limits the depth of quantum circuits in the NISQ era, there was an interesting open question (prior to our work) as to whether one could still train the parameters of a variational ansatz in the high decoherence limit. This question was especially important for VQAs for optimization, compiling, and linear systems, which are applications that do not require accurate estimation of cost functions on the quantum computer. Our work essentially provides a negative answer to this question. Naturally, important future work will involve extending our results to more general (e.g., non-unital) noise models, and numerically testing the tightness of our bounds. Moreover, our work emphasizes the importance of short-depth variational ansatzes. Hence a crucial research direction for the success of VQAs will be the development of methods to reduce ansatz depth.

\section{Methods}

\subsection{Special cases of our ansatz}

Here we discuss how the the QAOA, the Hardware Efficient ansatz, and the UCC ansatz fit into the framework as described in the general framework subsection.


1.~Quantum Alternating Operator Ansatz. The QAOA can be understood as a discretized adiabatic transformation where the goal is to prepare the ground state of a given Hamiltonian $H_P$. The order $p$ of the Trotterization determines the solution precision and the circuit depth. Given an initial state $\ket{\vec{s}}$, usually the linear superposition of all elements of the computational basis $\ket{\vec{s}}=\ket{+}^{\otimes n}$, the ansatz corresponds to the sequential application of  two unitaries $U_P(\gamma_l)=e^{-i \gamma_l H_P}$ and $U_M(\beta_l)=e^{-i \beta_l H_M}$. These alternating unitaries are  usually known as the problem   and mixer unitary, respectively. Here $\vec{\gamma}=\{\gamma_k\}_{l=1}^L$ and $\vec{\beta}=\{\beta_k\}_{l=1}^L$ are vectors of variational parameters which determine how long each unitary is applied and which must be optimized to minimize the cost function $C$,  defined as the  expectation value
 \begin{equation}
 C=\matl{\vec{\gamma},\vec{\beta}}{H_P}{\vec{\gamma},\vec{\beta}}=\Tr[H_P\dya{\vec{\gamma},\vec{\beta}}]\,,
 \end{equation}
 where $\ket{\vec{\gamma},\vec{\beta}}=U(\vec{\gamma},\vec{\beta})\ket{\vec{s}}$ is the QAOA variational state, and where $U(\vec{\gamma},\vec{\beta})$ is given by~\eqref{eq:QAOA}. In Fig.~\ref{fig:ansatz}(a) we depict the circuit description of a QAOA ansatz for a specific Hamiltonian where $k_P=6$.

\begin{figure}[t]
    \centering
    \includegraphics[width=1\columnwidth]{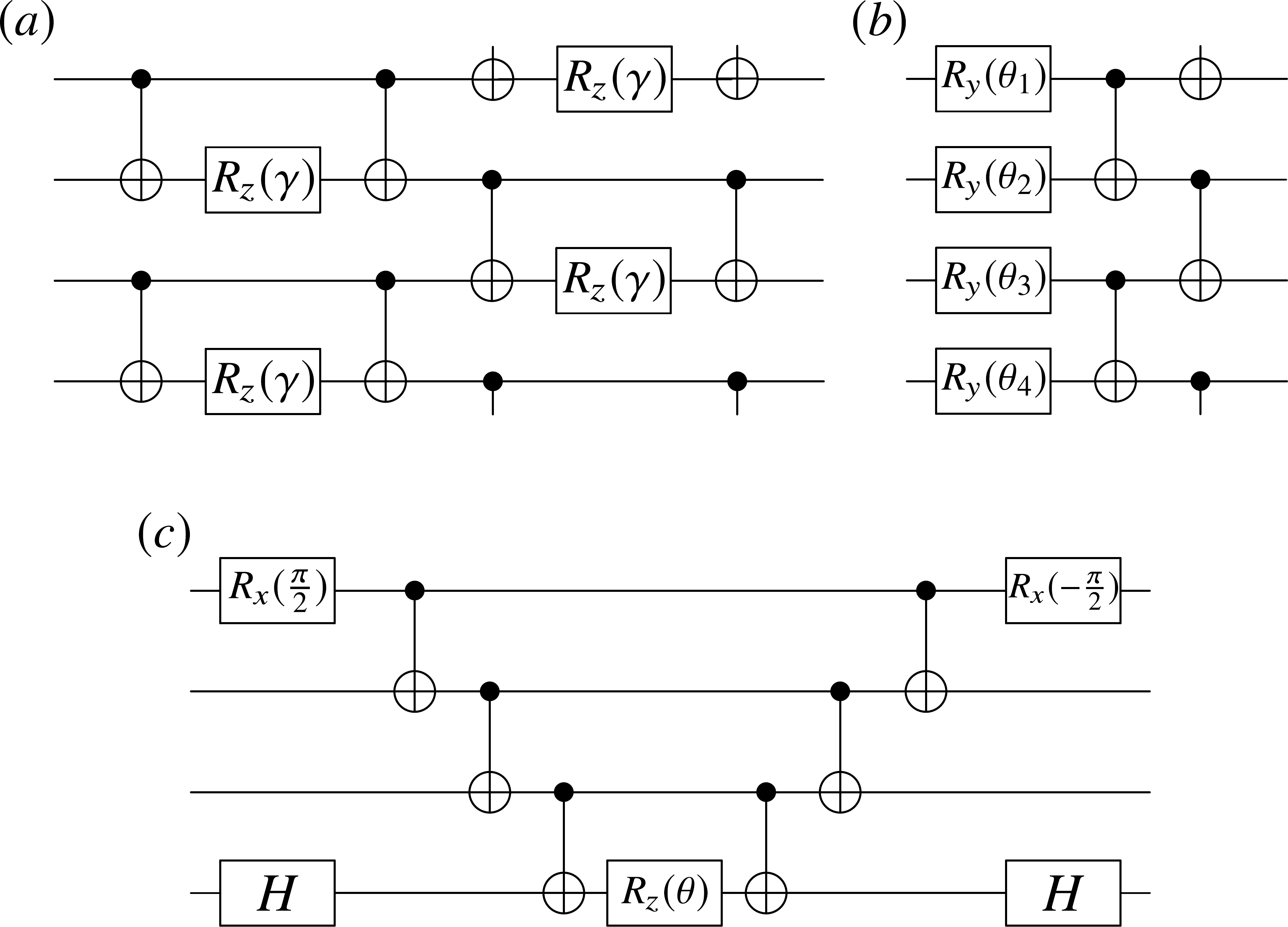}
    \caption{\textbf{Special cases of our general  ansatz.} (a) QAOA problem  unitary $e^{-i \gamma H_P}$ for the ring-of-disagrees MaxCut problem, with  Hamiltonian $H_P=\frac{1}{2}\sum_j Z_j Z_{j+1}$. (b) Hardware Efficient ansatz composed of CNOTs and single qubit rotations around the $y$-axis $R_y(\theta)$. (c) Unitary for the exponential $e^{-i  \theta  Y_{1}Z_{2}Z_{3}X_{4}}$. This type of circuit is a representative component of the UCC ansatz.     }
    \label{fig:ansatz}
\end{figure}

2.~Hardware Efficient Ansatz. The goal of the Hardware Efficient ansatz is to reduce the gate overhead (and hence the circuit depth) which arises when implementing a general unitary as in~\eqref{eq:layerunitary}. Hence, when employing a specific quantum hardware the parametrized gates $e^{-i \theta_{l m}H_{l m}}$ and the unparametrized gates $W_{l m}$ are taken from a gate alphabet composed of native gates to that hardware. Figure~\ref{fig:ansatz}(b) shows an example of a  Hardware Efficient ansatz where the gate alphabet is composed of rotations around the $y$ axis and of CNOTs.

3.~Unitary Coupled Cluster Ansatz. This ansatz is employed to estimate the ground state energy of the molecular Hamiltonian. In the second quantization, and within the Born-Oppenheimer approximation, the molecular Hamiltonian of a system of $M_e$ electrons can be expressed as:
$H = \sum_{pq} h_{pq} a\ad_p a_q + \frac{1}{2} \sum_{pqrs} h_{pqrs} a\ad_p a\ad_q a_r a_s$, where $\{a_p\ad\}$  ($\{a_q\}$) are Fermionic creation (annihilation) operators. Here, $h_{pq}$ and $h_{pqrs}$ respectively correspond to the so-called one- and two-electron integrals~\cite{cao2019quantum,mcardle2020quantum}. The ground state energy of $H$ can be  estimated with the VQE algorithm by preparing a reference state, normally taken to be the Hartree-Fock (HF) mean-field state $\ket{\psi_0}$, and acting on it with a parametrized UCC ansatz. 

The  action of a UCC ansatz with single ($T_1$) and double ($T_2$) excitations is given by $\ket{\psi} = \exp(T-T\ad)\ket{\psi_0}$, where $T=T_1+T_2$, and where
\begin{align}
T_{1}= \sum_{\substack{i \in \mathrm{occ}\\
a \in \mathrm{vir}}} t_{i}^{a} a_{a}^{\dagger} a_{i}, \quad 
T_{2}= \sum_{\substack{i, j \in \mathrm{occ}\\
a,b\in \mathrm{vir}}} t_{i, j}^{a,b} a_{a}^{\dagger} a_{b}^{\dagger} a_{j} a_{i}\,.
\end{align}
Here the $i$ and $j$ indices range over ``occupied'' orbitals whereas the $a$ and $b$ indices range over ``virtual'' orbitals \cite{cao2019quantum,mcardle2020quantum}. The coefficients $t^a_i$ and $t^{a, b}_{i,j}$ are called coupled cluster amplitudes. For simplicity, we denote these amplitudes $\{t^a_i, t^{a, b}_{i,j}\}$ as
$\{\theta_{l m}\}$. Similarly, by denoting the excitation operators \{$a_{a}^{\dagger} a_{i}$, $a_{a}^{\dagger} a_{b}^{\dagger} a_{j} a_{i}\}$ as $\{\tau_{l m}\}$, the UCC ansatz can be written in a compact form as $U(\vec{\theta}) = e^{\sum_{l m} \theta_{l m}(\tau_{l m}-\tau\ad_{l m})}$. In order to implement $U(\vec{\theta})$   one maps the fermionic operators to spin operators by means of the Jordan-Wigner or the Bravyi-Kitaev transformations \cite{ortiz2001quantum,bravyi2002fermionic}, which allows us to write $(\tau_{l m} - \tau_{l m}\ad) = i \sum_i \mu^i_{l m}  \sigma^i_n$. Then,  from a first-order Trotterization we obtain~\eqref{eq:ucc}. Here, $\mu^i_{l m} \in \{0,\pm1\}$. In Fig.~\ref{fig:ansatz}(c) we depict the circuit description of a representative component of the UCC ansatz.

\subsection{Proof of Theorem~\ref{thm1}}

Here we outline the proof for our main result on Noise-Induced Barren Plateaus. We refer the reader to the Supplementary Note \hyperref[sec:thm1proof-supp]{2} for additional details. We note that Lemma \ref{lemma1} and Remark \ref{remark1} follow from similar steps and their proofs are detailed in Supplementary Notes~\hyperref[sec:prooflemma]{3} and~\hyperref[sec:proofremark]{4} respectively. Moreover, we remark that Corollaries \ref{cor1}, \ref{cor:qaoa} and \ref{cor:ucc} follow in a straightforward manner from a direct application of Theorem~\ref{thm1} and Remark~\ref{remark1}.

Throughout our calculations we find it useful to use the expansion of operators in the Pauli tensor product basis. Given an $n$-qubit Hermitian operator $\Lambda$, one can always consider the decomposition
\begin{equation}\label{eq:paulibreakdown}
    \Lambda = \lambda_0 {\id}^{\otimes n} + \vec{\lambda}\cdot\vec{\sigma}_n\,, 
\end{equation}
where $\lambda_0 \in \mathbb{R}$ and $\vec{\lambda} \in \mathbb{R}^{4^n-1}$. Note that here we redefine the vector of Pauli strings $\vec{\sigma}_n$ as a vector of length  $4^n-1$ which excludes ${\id}^{\otimes n}$. 

Central to our proof is to understand how operators are mapped by concatenations of unitary transformations and noise channels. We do this through two lenses. First, given an operator $\Lambda$ we investigate how various $\ell_p$-norms of $\vec{\lambda}$ are related at different points in the evolution. Such quantities are well suited to study in our setting as we can use the transfer matrix formalism in the Pauli basis, that is, to represent a channel $\mathcal{N}$ with the matrix $(T_\mathcal{N})_{ij} = \frac{1}{2^n}\Tr\big[\sigma^i_n\, \mathcal{N}(\sigma^j_n)\big]$. Indeed, we see that the noise model in \eqref{eq:noisemodel} has a diagonal Pauli transfer matrix, which motivates this choice of attack. The second quantity we use is the sandwiched 2-R\'enyi relative entropy $D_2\big(\rho\big\Vert \id^{\otimes n} /2^n \big)$ between a state $\rho$ and the maximally mixed state. This is also useful to study due to the strong data processing inequality in Ref.~\cite{hirche2020contraction} which quantifies how noise maps $\rho$ closer to the maximally mixed state. 

Let us now present some lemmas that reflect these two perspectives. The action of the noise in \eqref{eq:noisemodel} on the operator $\Lambda$ is to map the elements of $\vec{\lambda}$ as $\lambda_i \xrightarrow[]{\,\mathcal{N}\,} \widetilde{\lambda}_i= q_X^{x(i)}q_Y^{y(i)}q_Z^{z(i)} \lambda_i$ where  $x(i)$, $y(i)$, and $z(i)$ respectively denote the number of $X,Y$, and $Z$ operators in the $i$-th Pauli string. Recall the definition $q=\sqrt{\max\{|q_X|,|q_Y|,|q_Z|\}}$. Since $x(i)+y(i)+z(i)\geq 1$,   
the inequality  $| \widetilde{\lambda}_i | \leq q^2 | {\lambda_i}| $ always holds for all $i$. Second, the action of a unitary is to preserve the size of the coefficients  measured in Euclidean norm $\|\vec{\lambda}\|_2$. This can be seen from the correspondence between $\|\vec{\lambda}\cdot \vec{\sigma}_n\|_2$ and $\|\vec{\lambda}\|_2$ and the unitary invariance of Schatten norms. Together, this already provides intuition that the effect of successive layers of unitaries and noise is to shrink the length of the Pauli coefficient vector. The second lemma we present is a consequence of a strong data-processing inequality of of the sandwiched 2-R\'enyi relative entropy of Ref.~\cite{hirche2020contraction}, from which we can show
\begin{equation}\label{eq:thmlemma-entropy}
    D_2\big(\WC^k(\rho)\big\Vert {\id^{\otimes n}} /2^n \big) \leq q^{2ck} D_2\big(\rho\big\Vert \id^{\otimes n} /2^n \big)\,,
\end{equation}
where $\WC^k$ denotes the concatenated channel of $k$ layers of unitary channels and noise channels $\NC$, $c=1/(2\ln 2)$ is a constant, and we note that $D_2\big(\rho\big\Vert \id^{\otimes n} /2^n \big)$ itself is always upper bounded by $n$ for any $n$-qubit quantum state $\rho$.

With this we have the main tools we present a sketch of the proof of a variant of Theorem \ref{thm1}. In order to analyze the partial derivative of the cost function
$ \partial_{lm}\widetilde{C}=\Tr\left[ O\,\partial_{l m}\,\rho_L  \right]$ 
we first note that the output state $\rho_L$ can be expressed as
\begin{align}
   \rho_L &= \left(\NC \circ \mathcal{W}_+ \circ  \mathcal{W}_- \right)( \rho_{0} )= \NC \circ\mathcal{W}_+ ( \rho_{l-} ) \,,
\end{align}
where $\rho_0$ is the input state, we denote $\rho_{l-}=\mathcal{W}_-(\rho_0)$, and 
\begin{align}
    \mathcal{W}_+ &= \mathcal{U}_L\circ \cdots \circ \mathcal{U}_{l+1} \circ \mathcal{N} \circ \mathcal{U}^{+}_{lm}\,, \\
    \mathcal{W}_- &= \mathcal{U}^{-}_{lm} \circ \mathcal{N} \circ \mathcal{U}_{l-1} \circ \cdots \circ \mathcal{N} \circ \mathcal{U}_1 \circ \mathcal{N} \,,
\end{align}
where $\mathcal{U}^{\pm}_{lm}$ are channels that implement the unitaries  ${U}^{-}_{lm} = \prod_{s\leq m} e^{-i \theta_{ls} H_{ls}} $ and ${U}^{+}_{lm} = \prod_{s>m} e^{-i \theta_{ls} H_{ls}} $ such that $U_{l}={U}^{+}_{lm}\cdot{U}^{-}_{lm}$. For simplicity of notation here we have omitted the parameter dependence on the concatenation of channels. It is straightforward to show that
\begin{equation}\label{eq:thm1lem-commutator}
    \partial_{lm} {\rho}_{l-}  = - i [H_{lm}, {\rho}_{l-}]\,.
\end{equation}

Using the tracial matrix H{\"o}lder's inequality \cite{baumgartner2011inequality}, we can write
\begin{align}
        \big|\partial_{l m} \widetilde{C}\big| &= \big|\mathrm{Tr}\left[O\,  \NC \circ \WC_+ (\partial_{l m}\,\rho_{l-}) \right]\big| \label{eq:gradient-methodology} \\
        &\leq \left\|O\right\|_\infty\, \left\| \NC \circ \WC_+ (\partial_{l m}\,\rho_{l-}) \right\|_1\,.
\end{align}
We can also consider the action of the adjoint map of $\NC$ on $O$, but for now proceed with a simpler proof which leads to a similar bound. We bound the second term by using \eqref{eq:thm1lem-commutator}, a lemma  on the correspondence between commutators and linear combinations of maps of \cite{li2017hybrid}, quantum Pinsker's inequality \cite{ohya2004quantum}, and \eqref{eq:thmlemma-entropy} to obtain $\left\| \NC \circ \WC_+ (\partial_{l m}\,\rho_{l-}) \right\|_1 \leq \sqrt{8\ln{2}}\, \|\vec{\eta}_{lm}\|_\infty\, n^{1/2} q^{c(L+1)}$. This gives the bound
\begin{equation}
    \big|\partial_{l m} \widetilde{C}\big| \leq \sqrt{8\ln{2}}\, \|O\|_{\infty}  \big\|\vec{\eta}_{lm}\big\|_\infty n^{1/2} q^{c(L+1)}\,, 
\end{equation}
which is essentially of the form of the bound in Theorem \ref{thm1}. We defer the reader to the Supplementary Information for the full proof. 

\subsection{Proof of Proposition \ref{prop1}}

Here we sketch the proof of Proposition \ref{prop1}, with additional details being presented in Supplementary Note \hyperref[sec:prop1si]{8}. 

We model measurement noise as a tensor product of independent local classical bit-flip channels, which mathematically corresponds to modifying the local POVM elements $P_0 = \dya{0}$ and $P_1 = \dya{1}$ as follows: 
\begin{align}
    P_0 = \dya{0} &\rightarrow \widetilde{P}_0 = \frac{1+q_M}{2} \dya{0} + \frac{1-q_M}{2} \dya{1} \,, \\
    P_1 = \dya{1} &\rightarrow \widetilde{P}_1 = \frac{1-q_M}{2} \dya{0} + \frac{1+q_M}{2} \dya{1} \,.
\end{align}
In turn, it follows that one can also model this measurement noise as a tensor product of local depolarizing channels with depolarizing probability $1 \geq (1-q_M)/2 \geq 0$, which we indicate by $\mathcal{N}_M$. The channel is applied directly to the measurement operator such that $\mathcal{N}_M(O)=\sum_i \omega^i\mathcal{N}_M(\sigma_n^i)=\widetilde{\vec{\omega}}\cdot \vec{\sigma}_n$. Here $\widetilde{\vec{\omega}}$ is a vector of coefficients $\widetilde{\omega}^i=q_M ^{w(i)}\omega^i$, where $w(i)=x(i) + y(i) + z(i)$ is the weight of the Pauli string. We recall that we have respectively defined $x(i)$, $y(i)$, $z(i)$ as the number of Pauli operators $X$, $Y$, and $Z$ in the $i$-th Pauli string.


Let us return to the partial derivative of the cost in \eqref{eq:gradient-methodology}. In the presence of measurement noise we then have 
\begin{align}
        \big|\partial_{l m} \widetilde{C}\big| &\leq \left\|\NC_M(O)\right\|_\infty\, \left\| \NC \circ \WC_+ (\partial_{l m}\,\rho_{l-}) \right\|_1 \,.
\end{align}
The term $\left\|\NC_M(O)\right\|_\infty$ can be bounded by the sum of noisy Pauli coefficients, each of which are dampened by factor $q_M^{w(i)}\leq q_M^{w}$. This gives an extra locality-dependent factor in the bound on the partial derivative:
    \begin{equation}
        |\partial_{l m} \widetilde{C}| \leq q_M^{w} F(n) \,.
    \end{equation}
    
An analogous reasoning leads to the following result for the concentration of the cost function:
    \begin{equation}\label{eq:prop-b12}
        \left\vert\widetilde{C} - \frac{1}{2^n}\Tr\,O\right\vert \leq q_M^{w} G(n) \,.
    \end{equation}

\subsection{Details of numerical implementations}

The noise model employed in our numerical simulations was  obtained by performing one- and two-qubit gate-set tomography~\cite{blume2017demonstration,Nielsen_2020} on the five-qubit IBM Q Ourense superconducting qubit device. The process matrices for each gate native to the device's  alphabet, and the state preparation and measurement noise are described in Ref.~\cite[Appendix B]{cincio2021machine}. In addition, the optimization for the MaxCut problems was performed using an optimizer based on the Nelder-Mead simplex method.

\subsection{Data availability}

Data generated and analyzed during the current study are available from the corresponding author upon reasonable request.

\subsection{Code availability}

Further implementation details are available from the authors upon request.

\section{Acknowledgements}

We thank Daniel Stilck Fran\c{c}a for helpful discussions and for pointing us to Ref.~\cite{muller2016relative}. Research presented in this article was supported by the Laboratory Directed Research and Development program of Los Alamos National Laboratory under project number 20190065DR. SW and EF acknowledge support from the U.S. Department of Energy (DOE) through a quantum computing program sponsored by the LANL Information Science \& Technology Institute. MC and AS were also supported by the Center for Nonlinear Studies at LANL. PJC also acknowledges support from the LANL ASC Beyond Moore's Law project. LC and PJC were also supported by the U.S. Department of Energy (DOE), Office of Science, Office of Advanced Scientific Computing Research, under the Quantum Computing Applications Team (QCAT) program. 

\section{Author contributions}
The project was conceived by PJC, MC, KS, and LC. Lemma 1 and Theorem 1 were proven by SW. Proposition 1 was proven by EF and MC. Corollaries 2 and 3 were proven by KS and SW. Numerical heuristics were run by LC. Implementation on quantum hardware was run by SW. The manuscript was written by SW, EF, KS, MC, AS, LC, and PJC. 

\section{Competing interests}
The authors declare no competing interests.

\bibliography{Ref.bib}

\onecolumngrid

\pagebreak

\appendix

\vspace{0.5in}

\setcounter{theorem}{0}

\begin{center}
	{\Large \bf Supplementary Information for \sl Noise-Induced Barren Plateaus in Variational Quantum Algorithms} 
\end{center}

 
In this Supplementary Information we provide proofs for the main results of the manuscript ``Noise-induced barren plateaus in variational quantum algorithms''. In Supplementary Note \hyperref[sec:pre-supp]{1} we first present some definitions and lemmas which will be useful in deriving our results. We point readers to \cite{nielsen2010, wilde2017, khatri2020principles} for additional background on relevant tools in quantum information theory. Then, in Supplementary Note \hyperref[sec:thm1proof-supp]{2} we present a detailed proof of our main result Theorem~\ref{thm1}. Supplementary Notes~\hyperref[sec:prooflemma]{3} and~\hyperref[sec:proofremark]{4} respectively contain the proofs for Lemma~\ref{lemma1} on cost concentration, and Remark \ref{rem:degparamsupp} on a generalization to correlated (degenerate) parameters. We present our proof for Remark \ref{remark:extension_noisemodel} on extensions to the noise model to $k$-local noise in Supplementary Note \hyperref[sec:remark2-supp]{5} and our proofs of Corollaries~\ref{cor:qaoa} and~\ref{cor:ucc} on application-specific results in Supplementary Note~\hyperref[sec:cor2sup]{6}. In Supplementary Note \hyperref[sec:remark:QML-supp]{7} we discuss our Remark \ref{remark:QML} on a construction where the cost function is summed over some dataset. Finally, the proof for Proposition~\ref{prop1} on measurement noise is detailed in Supplementary Note \hyperref[sec:prop1si]{8}.

\section*{Supplementary Note 1 - Preliminaries}\label{sec:pre-supp}

\subsection{Definitions}

\textbf{Quantum states.} Given some choice of Hilbert space $\mathcal{H}$, we denote the set of density operators as $\mathcal{S}(\mathcal{H})$. \bigskip

\textbf{Pauli expansion.} We note that one can always expand  $H_{l m}$ and $O$ in the Pauli basis as
\begin{align}\label{eq:HamiltonianDefSupp}
    H_{l m} &= \sum_{i} \eta^{i}_{l m}\sigma_n^{i} = c_{l m}^0\sigma_n^0 + \vec{\eta}_{l m} \cdot \vec{\sigma}_n\,,\\
    O &=\sum_{i}\omega^{i}\sigma_n^{i} =\omega^0\sigma_n^0 + \vec{\omega} \cdot \vec{\sigma}_n \label{eq:OSupp}\,.
\end{align}
where now  $\sigma_n^{i}\in \{\id,X,Y,Z\}^{\otimes n} \text{\textbackslash} \{\id^{\otimes n} \}$    length-$n$ Pauli strings. Here we remark that for the sake of simplicity we have made a subtle change in notation as now $\sigma_n^0 = \id^{\otimes n}$ is treated on a separate footing. With this notation,  $\vec{\sigma}_n, \vec{\eta}_{l m}, \vec{\omega}$ are real vectors of length $4^{n}-1$ and run over indices $i\in [4^{n}-1]$. Moreover, we recall that we have defined $N_{l m}=\vert \vec{\eta}_{l m} \vert$, and $N_O=\vert \vec{\omega} \vert$ as the number of non-zero elements in each respective vector. Furthermore, note that we can always set $\omega^0=0$ and $\eta^0_{l m} = 0$ for all $l m$. This does not lose us generality in our setting as a non-zero $\omega^0$ corresponds to a trivial measurement, while a  non-zero $\eta^0_{l m}$ simply leads to a different choice in the Hamiltonian normalization.  
\bigskip

\textbf{Vector norms.} In what follows we use the usual definitions of the $p$-norms such that $\|\vec{a}\|_\infty \equiv \max_i |a_i|$ is the largest element of vector $\vec{a}$ and $\|\vec{a}\|_2 \equiv \sqrt{\sum_i |a_i|^2}$ is the Euclidean norm.

\bigskip

\textbf{Setting for our analysis.} As shown in  Fig.~\ref{fig:suppcircuit} we break down the circuit into $L$ unitaries preceded and followed by noisy channels acting on all qubits. Let $\rho_0$ and $\rho_l$ respectively denote the input state and the state obtained after the $l$-th unitary. Let $\mathcal{N}=\mathcal{N}_1\otimes\cdots\otimes\mathcal{N}_n$ denote the $n$-qubit noise channel. Then the noisy cost function $\widetilde{C}$, defined as the expectation value of an operator $O$, can be represented as follows: 
\begin{align}\label{eq:noisycostsupp}
    \widetilde{C} &= \mathrm{Tr}\Big[ O\; \big(\mathcal{N}\circ \mathcal{U}_L(\vec{\theta}_L) \circ \mathcal{N} \circ \cdots \circ \mathcal{U}_2(\vec{\theta}_2) \circ \mathcal{N} \circ \mathcal{U}_1(\vec{\theta}_1) \circ \mathcal{N} \big) (\rho_0) \Big]\,,
\end{align}
where the $l$-th unitary channel $\mathcal{U}_l(\vec{\theta}_l)$ implements the unitary operator
\begin{equation} \label{eq:ansatzSupp}
    U_l(\thv_l)= \prod_{m} e^{-i \theta_{l m} H_{l m}} W_{l m}\,.
\end{equation}
Here we recall that $\thv_l=\{\theta_{l m}\}$ are continuous parameters and $W_{l m}$ denote unparameterized gates.

\bigskip

\textbf{Noise model.}
We consider a noise model where local Pauli noise channels $\mathcal{N}_j$ act on each qubit $j$ before and after each unitary $U_l(\thv_l)$. The action of $\mathcal{N}_j$ on a local Pauli operator $\sigma\in\{X,Y,Z\}$  can be expressed as
\begin{equation}\label{eq:noisemodelSupp}
    \mathcal{N}_j(\sigma)=q_{\sigma}\sigma\,,
\end{equation}
where $-1< q_X,q_Y,q_Z<1$. Here, we characterize the noise strength with a single parameter
\begin{equation} \label{eq:qmax}
    q=\max\{|q_X|,|q_Y|,|q_Z|\}\,.
\end{equation}

\bigskip

\textbf{Representation of the quantum state.} Here we will use the Pauli representation of an $n$ qubit state
\begin{align}
    \rho &= \frac{1}{2^n}\Big( \id^{\otimes n} + \vec{a}\cdot \vec{\sigma}_n \Big) \,,
\end{align}
where $a_i = \langle \sigma_n^i \rangle = \mathrm{Tr}[\rho\, \sigma_n^i]$ for all $i \in [4^{n}-1]$. The state $\rho$ can then be represented by a vector $\vec{a}$ of length $(4^{n}-1)$, with elements $a_i$, which we will refer to as the Pauli coefficients.

Recalling that $\rho_l$ is the state obtained after the application of the $l$-th unitary, we  employ the notation $a^{(l)}_i$ for its Pauli coefficients. That is, we explicit write the output of layer $l$ as 
\begin{align}
    \rho_l &= \frac{1}{2^n}( \id^{\otimes n} + \sum^{4^{n}-1}_{i=1} a^{(l)}_i \, \sigma_n^i ) = \frac{1}{2^n}( \id^{\otimes n} +  \vec{a}^{(l)} {\cdot}\, \vec{\sigma}_n ). \label{eq:paulibreakdown-supp}
\end{align}

\begin{figure}
    \centering
    \includegraphics[totalheight=0.25\textheight,trim={4cm 8.5cm 4cm 7cm},clip]{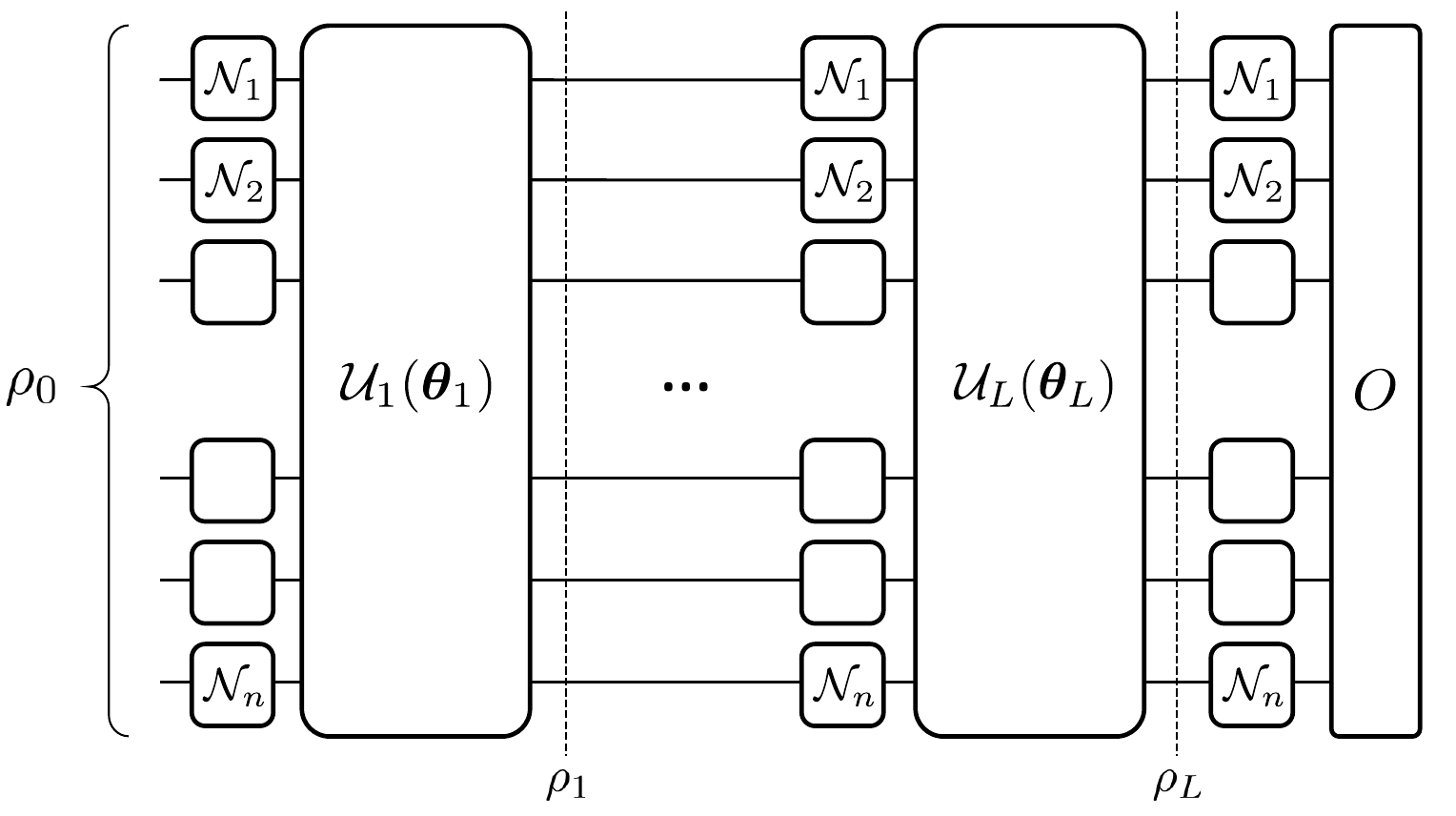} \caption{\textbf{Setting for our analysis.} An $n$-qubit input state $\rho_0$ is sent through a variational ansatz $U(\thv)$ composed of $L$ unitary layers $U_l(\thv_l)$ sequentially acting according to Eq.~\eqref{eq:ansatzSupp}. Here,  $\mathcal{U}_l$ denotes the quantum channel that implements the unitary $U_l(\thv_l)$.   The parameters in the ansatz $\thv=\{\thv_l\}_{l=1}^L$ are trained to minimize a cost function that is expressed as the expectation value of an operator $O$ as in Eq.~\eqref{eq:noisycostsupp}. We consider a noise model where local Pauli noise channels $\mathcal{N}_j$ act on each qubit $j$ before and after each unitary. We denote the state obtained after $l$ applications of noise followed by unitary as $\rho_l$. } \label{fig:suppcircuit}
\end{figure}

\subsection{Useful lemmas}
Here we present some supplementary lemmas which will be useful in deriving and understanding our main results.  

For convenience we start by quoting two standard results and point the reader to \cite{wilde2017} for further details.

\begin{suplemma}[Tracial matrix H\"older's inequality \cite{baumgartner2011inequality}]\label{suplem:pinsker} 
Consider two $d\times d$ matrices $A$ and $B$. Then we have
\begin{equation}
    \left|\operatorname{Tr} A^{\dag} B\right| \leq\|A\|_{r}\,\|B\|_{s}, 
\end{equation}
for all $1 \leq r, s \leq \infty$ such that $ \frac{1}{r} + \frac{1}{s} = 1$.
\end{suplemma}

\begin{suplemma}[Pinsker's inequality \cite{ohya2004quantum}] 
Consider two quantum states $\rho,\sigma \in \mathcal{S}(\mathcal{H})$. Then, the quantum relative entropy $D(\rho \| \sigma)$ is lower bounded as
\begin{equation}
    D(\rho \| \sigma) \geq \frac{1}{2 \ln 2}\|\rho-\sigma\|_{1}^{2}\,.
\end{equation}
\end{suplemma}

Now we present a collection of lemmas on the effect of noise and unitary operations on quantum states and their Pauli coefficients.

\begin{suplemma}[Pauli coefficients under unitary transformations] \label{suplem:unitaries}
   Let $\Lambda$ be an $n$-qubit operator whose Pauli basis decomposition is 
   \begin{align}\label{eq:paulidecomposition_supp}
        \Lambda =  \lambda_0 \emph{\id}^{\otimes n} + \vec{\lambda}\cdot\vec{\sigma}_n,
        \end{align}
    where $\lambda_0 \in \mathbb{R}$ and $\vec{\lambda} \in \mathbb{R}^{4^{n}-1}$. Then $\|\vec{\lambda}\cdot \vec{\sigma}_n\|_p$ is invariant under the unitary transformation $\Lambda \rightarrow U\Lambda U^\dag$ for any unitary operator $U$. In particular, this also implies $\|\vec{\lambda}\|_2$ is invariant under unitary transformations.       
\end{suplemma}
\begin{proof}
    Denote the new Pauli coefficients after unitary transformation as $\vec{\lambda}'$. We note the first term in Eq.~\eqref{eq:paulidecomposition_supp} is invariant under such transformations. Thus
    \begin{equation}
        \vec{\lambda}'\cdot \vec{\sigma}_n = U(\vec{\lambda}\cdot \vec{\sigma}_n)U^\dag
    \end{equation}
    and so unitary invariance of $\|\vec{\lambda}\cdot \vec{\sigma}_n\|_p$ follows from the unitary invariance of Schatten norms. To see the unitary invariance of $\|\vec{\lambda}\|_2$ note that for $p=2$ we have
    \begin{align}
        \|\vec{\lambda}\cdot \vec{\sigma}_n\|_2 &= \sqrt{\Tr\left[ (\vec{\lambda}\cdot \vec{\sigma}_n)^2 \right] } \\
        &= \sqrt{\Tr\big[\vec{\lambda}\cdot\vec{\lambda\,}{\id^{\otimes n}}\big]} \\
        &= {2^{n/2}}\|\vec{\lambda}\|_2\,\
    \end{align}
    and thus unitary invariance of $\|\vec{\lambda}\cdot \vec{\sigma}_n\|_2$ implies unitary invariance of $\|\vec{\lambda}\|_2$.
\end{proof}

\begin{suplemma}[Pauli coefficients under noise] \label{suplem:noise}
    Consider an operator $\Lambda$ of the form \eqref{eq:paulidecomposition_supp}. Under the action of a noise channel $\mathcal{N}$ of the form in Eq.~\eqref{eq:noisemodelSupp} on a single Pauli string $\sigma^i_n$ we have
    \begin{equation}
        |\widetilde{\lambda}_i| \leq q |\lambda_i|\,,
    \end{equation}
    where we define $\widetilde{\lambda}_i$ as the coefficient such that $\widetilde{\lambda}_i\, \sigma^i_n = \mathcal{N}(\lambda_i\sigma^i_n)$.
\end{suplemma}
\begin{proof}
    The effect of a single layer of noise can be expressed as follows:
    \begin{equation}\label{eq:pauli-coefficient-decay}
        \lambda_i \xrightarrow[]{\,\mathcal{N}\,} q_X^{x(i)}q_Y^{y(i)}q_Z^{z(i)} \lambda_i \,,
    \end{equation}
    for all $i \in [4^n-1]$, where $x(i)+y(i)+z(i)\leq n$ is the number of non-identity terms in the $i$-th Pauli string.
    Noting that $x(i)+ y(i)+z(i)\geq 1 \,\;\forall i$ and using Eq.~\eqref{eq:qmax}, we obtain the desired statement. 
\end{proof}

\begin{suplemma}[Relative entropy contraction; \textit{Müller-Hermes/França/Wolf} \cite{muller2016relative}\textit{, Theorem 6.1}]\label{suplem:relentropy}
    Consider a channel
    \begin{equation}
       \mathcal{W} = \mathcal{U}_k\circ \mathcal{N}\circ \cdots \circ \mathcal{N} \circ \mathcal{U}_2 \circ \mathcal{N} \circ \mathcal{U}_1\circ\mathcal{N}
    \end{equation}
    that consists of $k$ noise channels $\NC= \NC_1 \otimes ... \otimes \NC_n$ where each $\NC$ is a depolarizing noise channel with depolarizing probability $p$, interleaved with unitary channels $\mathcal{U}_i$.
    Denote the relative entropy as $D\big(\cdot\big\Vert\cdot\big)$. We have 
    \begin{equation}
      D\bigg(\mathcal{W}(\rho)\bigg\Vert \frac{\emph{\id}^{\otimes n}}{2^n} \bigg) \leq (1-p)^{2k} D\bigg(\rho\bigg\Vert \frac{\emph{\id}^{\otimes n}}{2^n} \bigg) \leq (1-p)^{2k} n.
    \end{equation}
\end{suplemma}

We remark that depolarizing noise of probability $p$ is a special case of our noise model with $q=(1-p)$. In this special case, we can use the above lemma to rederive slightly stronger bounds for Theorem \ref{thm1} and Proposition \ref{prop1}. For the vast majority of our results, we shall instead use the following lemma, which extends the scope to a more general class of Pauli noise channels, at a cost of a slightly weaker bound. 

\begin{suplemma}[Sandwiched 2-R\'enyi relative entropy contraction] \label{suplem:2renyientropy}
Consider a single instance of the noise channel $\NC = \NC_1 \otimes ... \otimes \NC_n$ where each local noise channel $\{\NC_j\}_{j=1}^n$ is a Pauli noise channel that satisfies \eqref{eq:noisemodelSupp}. Then, we have
\begin{equation}\label{eq:2renyi}
    D_{2}\!\left(\NC(\rho) \Big\| \frac{\emph{\id}^{\otimes n}}{2^n}\right) \leq q^{2c} D_{2}\!\left(\rho \Big\| \frac{\emph{\id}^{\otimes n}}{2^n}\right)\,.
\end{equation}
where $D_2\big(\cdot\big\Vert\cdot\big)$ denotes the sandwiched 2-R\'enyi relative entropy, and $c=1/(2\ln 2)$ is a constant.
\end{suplemma}

\begin{proof}
This lemma comes as a direct consequence of Corollary 5.6 of Ref.~\cite{hirche2020contraction}. Let us first restate the result for convenience: For some density operator $\sigma$ and $p>0$ consider the channel $\AC_{p,\sigma}(\cdot)=p (\cdot)+(1-p) \sigma$. Suppose that some other channel $\BC$ satisfies
\begin{equation}
    \left\|\Gamma_{\BC(\sigma)}^{-\frac{1}{2}} \circ \BC \circ \AC_{p,\sigma}^{-1} \circ \Gamma_{\sigma}^{\frac{1}{2}}\right\|_{2 \rightarrow 2} \leq 1\,
\end{equation}
where $\AC_{p,\sigma}^{-1}$ denotes the inverse map of $\AC_{p,\sigma}$ and $\Gamma_{\sigma}^{p}$ denotes the map $\Gamma_{\sigma}^{p}(\cdot)=\sigma^{\frac{p}{2}} (\cdot) \sigma^{\frac{p}{2}}$. Then, for all states $\rho$, 
\begin{equation}\label{eq:2renyi}
    D_{2}\!\left(\BC^{\otimes n}(\rho) \| \BC^{\otimes n}(\sigma^{\otimes n})\right) \leq \alpha(p, \sigma) D_{2}\!\left(\rho \| \sigma^{\otimes n}\right)
\end{equation}
where $\alpha(p, \sigma)=\exp \left(2\left(1-\left\|\sigma^{-1}\right\|^{-1}\right) \frac{\log (p)}{\log \left(\left\|\sigma^{-1}\right\|\right)}\right)$. If one chooses $\AC_{p,\sigma}$ to be the depolarizing channel $\DC_{p_d}$ with depolarizing probability $p_d$, then $\sigma = \id/2$ and $p=(1-p_d)$. 
Eq.~\eqref{eq:2renyi} implies that if some qubit channel $\BC$ satisfies
\begin{equation}\label{eq:ball}
    \left\|\BC \circ \DC_{p_d}^{-1} \right\|_{2 \rightarrow 2} \leq 1\,.
\end{equation}
then for any $n$-qubit state $\rho$ we have
\begin{equation}\label{eq:2renyi2}
    D_{2}\!\left(\BC^{\otimes n}(\rho) \Big\| \frac{{\id}^{\otimes n}}{2^n}\right) \leq (1-p_d)^{1/\ln 2} D_{2}\!\left(\rho \Big\| \frac{{\id}^{\otimes n}}{2^n}\right)\,.
\end{equation}

Now suppose that $\BC$ is the qubit Pauli noise channel $\BC$ as defined in \eqref{eq:noisemodelSupp}. We can explicitly write the condition \eqref{eq:ball} as
\begin{equation}\label{eq:ball2}
    \sup_{X\neq 0} \frac{\|\BC\circ\DC_{p_d}^{-1}(X) \|_2}{\|X \|_2} \leq 1.
\end{equation}
We note that the superoperator (Pauli transfer matrix) of the concatenated channel $\BC\circ\DC_{p_d}^{-1}$ is diagonal with diagonal entries $(1,\frac{q_x}{1-p_d},\frac{q_y}{1-p_d},\frac{q_z}{1-p_d})$. Consider an arbitrary complex matrix $X$ decomposed in the Pauli basis as $X = a\id + \Vec{b}\cdot\Vec{\sigma}$, where $\Vec{\sigma}$ is the vector of Pauli matrices and $\Vec{b}$ is a vector of complex coefficients. Then one can verify 
\begin{align}
    \|X \|_2 &= \sqrt{2}\sqrt{|a|^2 + \textstyle\sum_i |b_i|^2}\,, \\
    \|\BC\circ\DC_p^{-1} (&X) \|_2 = \sqrt{2}\sqrt{|a|^2 + \textstyle\sum_i\left(\frac{q_i}{1-p_d}\right)^2|b_i|^2} \,,
\end{align}
where the second equality is obtained by reading off the diagonal entries of the superoperator of $\BC\circ\DC_{p_d}^{-1}$. In order to satisfy condition \eqref{eq:ball2}, one can pick
\begin{equation}
    1-p_d=\max_{i \in \{X,Y,Z\}} |q_i|\,.
\end{equation}
Thus, by denoting $q = \max_{i \in \{X,Y,Z\}} |q_i|$ and inspecting \eqref{eq:2renyi2} we obtain the result as required.
\end{proof}

\begin{suplemma}[Pauli coefficients through noisy circuit]\label{suplem:noise+unitaries}

Consider the output state of a depth-$l$ noisy circuit as in Eq.~\eqref{eq:paulibreakdown}, with vector of Pauli coefficients $\vec{a}^{(l)}$. This satisfies
\begin{equation}
    \big\|\vec{a}^{(l)} {\cdot}\, \vec{\sigma}_n\big\|_1 \leq 2^n \sqrt{2 \ln{2} \cdot n}\, q^{cl}\,,
\end{equation}
where $c=1/(2\ln 2)$ is a constant.
    
\end{suplemma}
\begin{proof}
We recall that we denote the state obtained at the end of the depth-$l$ circuit as the form in Eq.~\eqref{eq:paulibreakdown-supp}. With this we can write
\begin{align}
    \frac{1}{2^n}\big\|\vec{a}^{(l)} {\cdot}\, \vec{\sigma}_n\big\|_1 &\leq \sqrt{2 \ln 2\cdot D\big(\rho_l \big\| \frac{\id}{2^n}\big)}\\
    &\leq \sqrt{2 \ln 2\cdot D_2\big(\rho_l \big\| \frac{\id}{2^n}\big)}\\
    &\leq \sqrt{2 \ln 2\cdot q^{2cl} D_2\big(\rho_0 \big\| \frac{\id}{2^n}\big)}\\
    &\leq \sqrt{2 \ln{2} \cdot  q^{2cl}\, n}\,,
\end{align}
where the first inequality is due to Pinsker's inequality (Supplementary Lemma \ref{suplem:pinsker}); the second inequality is due to the monotonicity of the sandwiched 2-R\'enyi relative entropy; the third inequality follows from $l$ applications of Supplementary Lemma \ref{suplem:2renyientropy} and unitary invariance of the sandwiched 2-R\'enyi relative entropy; and in the final inequality we use a generic upper bound of the sandwiched 2-R\'enyi relative entropy.
\end{proof}

We note that as $\rho_l - \id/{2^n} = \vec{a}^{(l)} {\cdot}\, \vec{\sigma}_n/2^n $, this result can be thought of as a statement on the concentration of the trace distance of $\rho_l$ from the maximally mixed state. Finally, we present a lemma on commutators, which we will use to derive our result for gradients.

\begin{suplemma}[Commutator with Hermitian self-inverse operators; \textit{Li et al.}~\cite{li2017hybrid}\textit{, Equation 4}]\label{lem:commutator}
For any Hermitian self-inverse  operator $P$ which generates unitary $U_{P}(\theta) = \exp(-i\theta P/2)$ and any operator $A$ we have
\begin{equation}
    [P,A] = i\left( U_{P}\!\left(\frac{\pi}{2}\right)A\; U_{P}\!\left(\frac{\pi}{2}\right)^{\dag} - U_{P}\!\left(-\frac{\pi}{2}\right)A\; U_{P}\!\left(-\frac{\pi}{2}\right)^{\dag} \right)\,.
\end{equation}
\end{suplemma}

\section*{Supplementary Note 2 - Proof of Theorem \ref{thm1}}\label{sec:thm1proof-supp}

Here we provide the proof for our main result of Theorem~\ref{thm1}, which we now recall for convenience.

\begin{theorem}[Upper bound on the partial derivative]\label{thm1si}
    Consider an $L$-layered ansatz as defined in Eq.~\eqref{eq:ansatzSupp}. Let $\theta_{l m}$ denote the trainable parameter corresponding to the Hamiltonian $H_{l m}$ in the unitary $U_l(\vec{\theta}_l)$ appearing in the ansatz. Suppose that local Pauli noise of the form in Eq.~\eqref{eq:noisemodelSupp} with noise parameter $q$ acts before and after each layer as in Fig.~\ref{fig:suppcircuit}. Then the following bound holds for the partial derivative of the noisy cost function
    \begin{equation}\label{eq:th1sib}
       |\partial_{m} \widetilde{C}|  \leq F(n)\,,
       \end{equation}
       where
      \begin{equation}\label{eq:bound-thmsupp}
         F(n)= \sqrt{8\ln{2}}\, N_O \big\|\vec{\omega}\big\|_\infty \big\|\vec{\eta}_{lm}\big\|_1 n^{1/2} q^{cL+1}  \,,
    \end{equation}
    and $\vec{\eta}_{lm}$ and  $\vec{\omega}$ are defined in Eq.~\eqref{eq:HamiltonianDef}, $N_O$ is the number of non-zero Pauli terms in $O$, and and $c=1/(2\ln 2)$ is a constant.
\end{theorem}

\begin{proof}
    We write the overall channel that the state undergoes before measurement as the concatenation of two channels:
    \begin{equation}
        \mathcal{N}\circ \mathcal{U}_L(\vec{\theta}_L)\circ \cdots \circ \mathcal{N} \circ \mathcal{U}_2(\vec{\theta}_2) \circ \mathcal{N} \circ \mathcal{U}_1(\vec{\theta}_1)\circ\mathcal{N} \, (\cdot) = \mathcal{N} \circ \mathcal{W}_{+} \circ \mathcal{W}_{-} (\cdot) \,,
    \end{equation}
    where 
    \begin{align}
         \mathcal{W}_{-} &= \mathcal{U}^{-}_{m}(\vec{\theta}_l) \circ \mathcal{N} \circ \mathcal{U}_{l-1}(\vec{\theta}_{l-1}) \circ \cdots \circ \mathcal{N} \circ \mathcal{U}_1(\vec{\theta}_1) \circ \mathcal{N}, \\
         \mathcal{W}_{+} &=  \mathcal{U}_L(\vec{\theta}_L) \circ \cdots \circ \mathcal{U}_{l+1}(\vec{\theta}_{l+1}) \circ \mathcal{N} \circ \mathcal{U}^{+}_{m}(\vec{\theta}_l) \,,
    \end{align}
    with $\mathcal{W}_{-}$ containing $l$ layers of noise channels, and $\mathcal{W}_{+}$ containing $L-l$ layers of noise channels.
    Here we have defined the unitary channels $\mathcal{U}^{-}_{m}(\vec{\theta}_l)$ and $\mathcal{U}^{+}_{m}(\vec{\theta}_l)$ that respectively correspond to the following unitaries:
    \begin{align}
        {U}^{-}_{m}(\vec{\theta}_l) = \prod_{s=1}^{m} e^{-i \theta_{ls} H_{ls}}\,, \quad \quad
        {U}^{+}_{m}(\vec{\theta}_l) = \prod_{s>m}^{} e^{-i \theta_{ls} H_{ls} }\,,
    \end{align}
such that ${U}^{+}_{m}(\vec{\theta}_l) { U}^{-}_{m}(\vec{\theta}_l) = {U}_{l}(\vec{\theta}_l)$.     
For simplicity of notation let us denote $\partial_{l m}  \widetilde{C} = \partial_{\theta_{l m}} \widetilde{C}$. We have
    \begin{equation} \label{eq:derivativecostSupp}
        \partial_{l m} \widetilde{C} = \mathrm{Tr}[O\, \partial_{l m}\rho_L]\,,
    \end{equation}
    with
    \begin{align}
        \partial_{l m} \,  \rho_L &= \partial_{l m} \, \big( \NC \circ \mathcal{W}_{+} \circ \mathcal{W}_{-} ({\rho}_{0}) \big) \\
        &= \NC \circ \mathcal{W}_{+}\big( \partial_{l m}\, {\rho}_{l-} \big) \,,
    \end{align}   
where we denote $\mathcal{W}_{-} ({\rho}_{0}) = {\rho}_{l-}$. Thus we can write the derivative of the noisy cost function as
    \begin{align}
        \big|\partial_{l m} \widetilde{C}\big| &= \big|\mathrm{Tr}\left[O\, \NC \circ \mathcal{W}_{+}(\partial_{l m}\, {\rho}_{l-}) \right]\big| \\
        &\leq \left\|\NC(O)\right\|_\infty\, \left\|\mathcal{W}_{+}\big( \partial_{l m}\, {\rho}_{l-} \big)\right\|_1\, \label{eq:Hölder_supp}
    \end{align}
where Eq.~\eqref{eq:Hölder_supp} comes from application of Hölder's inequality, and we have used the fact that $\NC$ is self-adjoint. We now upper bound both terms individually.

The first term can be bounded as 
    \begin{align}
         \big\|\mathcal{N}(O)\big\|_\infty &=  \big\|\mathcal{N}(\vec{\omega}\cdot\vec{\sigma}_n)\big\|_\infty \\
         &\leq N_O \max_i\big\|\mathcal{N}(\omega^i \sigma^i_n)\big\|_\infty\\
         &\leq N_O\, q \|\vec{\omega}\|_\infty\,. \label{eq:thmW}
    \end{align}
    The first inequality is due to the triangle inequality, and a maximization over a sum of terms; the second inequality follows from Supplementary Lemma \ref{suplem:noise}, and by noting that the eigenvalues of Pauli matrices are $\{+1,-1\}$. 

Second, let us upper bound on the $1$-norm of $\mathcal{W}_{+}\big( \partial_{l m}\, {\rho}_{l-} \big)$. We have
    \begin{align}
         \partial_{l m} \, {\rho}_{l-} &=  -iH_{l m} {\rho}_{l-} +  i {\rho}_{l-} H_{l m}\\
         &=  -i \big[H_{l m} \,,\,  {\rho}_{l-} \big]  \\
         &=  -i \sum_{j} \big[\eta_{lm}^j \sigma^j_n \,,\,  {\rho}_{l-} \big]  \\
         & = \sum_{j} \eta_{lm}^j \big[ \sigma^j_n \,,\,  {\rho}_{l-} \big] \\
         & = \sum_{j} \eta_{lm}^j \left( \VC_{\sigma^j_n}( {\rho}_{l-}) - \VC_{\sigma^j_n}^{\dag} ({\rho}_{l-}) \right)\,,
    \end{align}
where in the final line we have used Supplementary Lemma \ref{lem:commutator} we denote $\VC_{\sigma_j}$ and $\VC^{\dag}_{\sigma_j}$ as the unitary channels corresponding to the unitary operators $\exp(-i\pi \sigma_j/4)$ and $\exp(i\pi \sigma_j/4)$ respectively.

This enables us to write
\begin{align}
    \Big\| \mathcal{W}_{+}\big( \partial_{l m}\, \rho_{l-} \big)\Big\|_1 &= \bigg\|  \sum_{j} \eta_{lm}^j \Big( \mathcal{W}_{+}\circ\VC_{\sigma_j} ( \rho_{l-}) - \mathcal{W}_{+}\circ\VC_{\sigma_j}^{\dag} ( \rho_{l-}) \Big) \bigg\|_1 \\
    &\leq \sum_{j} \big| \eta_{lm}^j \big| \left(  \Big\|  \mathcal{W}_{+}\circ\VC_{\sigma_j} ( \rho_{l-}) - \frac{\id}{2^n} \Big\|_1 + \Big\| \mathcal{W}_{+}\circ\VC_{\sigma_j}^{\dag} ( \rho_{l-}) - \frac{\id}{2^n}  \Big\|_1 \right) \\
    &\leq \big\| \vec{\eta}_{lm}\big\|_1 \max_j  \left(  \Big\|  \mathcal{W}_{+}\circ\VC_{\sigma_j} \big(\frac{1}{2^n}{\vec{a}}^{(l-)}{\cdot}\,\vec{\sigma}_n\big) \Big\|_1 + \Big\| \mathcal{W}_{+}\circ\VC_{\sigma_j}^{\dag} \big(\frac{1}{2^n}{\vec{a}}^{(l-)}{\cdot}\,\vec{\sigma}_n\big)  \Big\|_1 \right) \\
    &= \big\| \vec{\eta}_{lm}\big\|_1 \max_j  \left(  \Big\|  \frac{1}{2^n}{\vec{a}}^{(L,j)}{\cdot}\,\vec{\sigma}_n \Big\|_1 + \Big\| \frac{1}{2^n}{\vec{a}}^{(L,j')}{\cdot}\,\vec{\sigma}_n \Big\|_1 \right)\\
    &\leq \sqrt{8\ln{2}} \,   \big\|\vec{\eta}_{lm}\big\|_1\, n^{1/2} q^{cL}\,, \label{eq:thmrho}
\end{align}
where in the first inequality we have added and subtracted $\id/2^n$ and implemented the triangle inequality; in the second inequality we have taken a maximization over $j$ and rewritten $\rho_{l-}$ in terms of its Pauli coefficients $\vec{a}^{(l-)}$; in the following equality we have denoted the Pauli coefficients of the output state of the depth $L$ circuit with gates $\VC_{\sigma_j}$ and $\VC_{\sigma_j}^{\dag}$ inserted in between as ${\vec{a}}^{(L,j)}$ and ${\vec{a}}^{(L,j')}$ respectively; and the final inequality is an implementation of Supplementary Lemma \ref{suplem:noise+unitaries}. We see that the concentration of the $1$-norm of the partial derivative here originates from the fact that it can be expressed as a linear combination of state concentrations of depth $L$ circuits.
 
Inserting \eqref{eq:thmW} and \eqref{eq:thmrho} into Eq.~\eqref{eq:Hölder_supp}, we finally obtain   
    \begin{equation} \label{eq:suppthmfinal}
        \big|\partial_{l m} \widetilde{C}\big| \leq \sqrt{8\ln{2}}\, N_O \big\|\vec{\omega}\big\|_\infty \big\|\vec{\eta}_{lm}\big\|_1 n^{1/2} q^{cL+1} \,,
    \end{equation}
as required.   
\end{proof}

\subsubsection{Stronger bound for low noise levels under more restrictive Pauli noise model}
We note that via alternative proof techniques one may obtain a similar bound to Theorem \ref{thm1} for a different class of local Pauli noise models, where the bound is stronger in the regime of low noise strength (i.e., large $q$) and relatively uniform local Pauli error probabilities (i.e., close to local depolarizing noise). The core idea is that certain qubit Pauli channels can be decomposed into a depolarizing channel with non-trivial depolarizing probability, followed by a different Pauli channel. 

Consider a unital Pauli channel $\PC_{p_x,p_y,p_z}$ whose action on qubit state $\rho$ takes the form
\begin{equation}\label{eq:pauli}
    \PC_{p_x,p_y,p_z}(\rho) = p_I\rho + p_x X\rho X + p_y Y\rho Y + p_z Z\rho Z\,. 
\end{equation}
where $(p_I, p_x,p_y,p_z)$ is a probability vector with each $p_I, p_x, p_y, p_z > 0$. It can then be checked (e.g.~by using the superoperator formalism) that 
\begin{equation}\label{eq:channel_decomposition}
    \PC_{p_x,p_y,p_z} = \PC_{p'_x,p'_y,p'_z} \circ \DC_{p}\,.
\end{equation}
is a valid decomposition of $\PC_{p_x,p_y,p_z}$, where $\DC_{p}$ is a depolarizing channel with depolarizing probability $p=4\min(p_I,p_x,p_y,p_z)$ and $\PC_{p'_x,p'_y,p'_z}$ is a Pauli channel with $p'_j = \frac{p_j-p/4}{1-p}$ for all $j \in \{I, X, Y, Z\}$ . The decomposition \eqref{eq:channel_decomposition} allows us to directly use Supplementary Lemma \ref{suplem:relentropy} along with the data-processing inequality to modify Supplementary Lemma \ref{suplem:noise+unitaries} to 
\begin{equation}
    \big\|\vec{a}^{(l)} {\cdot}\, \vec{\sigma}_n\big\|_1 \leq 2^n \sqrt{2 \ln{2}}\, n^{1/2} \hat{q}^{l}\,.
\end{equation}
where $\hat{q}=1-4\min(p_I, p_x, p_y, p_z)\geq q$, with equality for depolarizing noise. This modifies the result of Theorem \ref{thm1si} to 
\begin{equation} 
    \big|\partial_{l m} \widetilde{C}\big| \leq \sqrt{8\ln{2}}\, N_O \big\|\vec{\omega}\big\|_\infty \big\|\vec{\eta}_{lm}\big\|_1 n^{1/2} q\, \hat{q}^{L} \,.
\end{equation}
We note that set of Pauli noise models for which $p>0$ in \eqref{eq:channel_decomposition} is a strict subset of those for which $q<1$ in \eqref{eq:noisemodel}. Thus, Theorem \ref{thm1si} gives a non-trivial bound for a more general class of Pauli noise models. 

\section*{Supplementary Note 3 - Proof of Lemma \ref{lemma1}}\label{sec:prooflemma}

In this supplementary note we provide a proof for Lemma \ref{lemma1}. We note that this Lemma is derived by employing  techniques similar to those used in deriving Theorem~\ref{thm1si} in  Supplementary Note \ref{sec:thm1proof-supp}. 

\setcounter{lemma}{0}

\begin{lemma}[Concentration of the cost function] \label{lemma3-supp}
  Consider an $L$ layer ansatz of the form in Eq.~\eqref{eq:ansatzSupp}. Suppose that local Pauli noise of the form of Eq.~\eqref{eq:noisemodelSupp} with noise strength $q$ acts before and after each layer as in Fig.~\ref{fig:suppcircuit}. Then, for a cost function $\widetilde{C}$ of the form in Eq.~\eqref{eq:noisycostsupp}, the following bound holds:
  \begin{equation}
      \left\vert \widetilde{C} - \frac{1}{2^n}\Tr[O]\right\vert \leq \sqrt{2 \ln{2}}\,   N_O\,\|\vec{\omega}\|_\infty\, n^{1/2} q^{cL+1}\,,
  \end{equation}
  where $N_O$ is the number of Pauli terms in the measurement operator $O$, $\|\vec{\omega}\|_\infty$ is the largest coefficient of those terms, and $c=1/(2\ln 2)$ is a constant.
\end{lemma}

\begin{proof}
We denote the overall channel that the state undergoes before measurement as $\mathcal{W}$ (a series of $L+1$ layers of noise channels, interleaved with unitary channels). We can write
\begin{align}
    \widetilde{C} &= \Tr[ O\, \mathcal{W}(\rho)]\\
    & = \frac{1}{2^n}\left(\Tr\big[O\big] + \Tr\big[O\,\NC({\vec{{a}}}^{(L)}\cdot \vec{\sigma}_n )\big]\right),
\end{align}
where we have used the Pauli decomposition \eqref{eq:paulibreakdown-supp}. This decomposition enables us to write 
\begin{align}
    \left\vert \widetilde{C} - \frac{1}{2^n} \Tr[O]\right\vert &= \bigg\vert  \Tr\big[O\,\NC({\vec{{a}}}^{(L)}\cdot \vec{\sigma}_n )\big]\bigg\vert \\
    &\leq  \big\|\NC(O)\big\|_{\infty} \big\|{\vec{{a}}}^{(L)}\cdot \vec{\sigma}_n \big\|_1 \\
    &\leq \sqrt{2 \ln{2}}\,   N_O\,\|\vec{\omega}\|_\infty\, n^{1/2} q^{cL+1} \,,
\end{align}
where the first inequality uses H\"older's inequality, and the second inequality comes from application of \eqref{eq:thmW} and Supplementary Lemma \ref{suplem:noise+unitaries}.  
\end{proof}

\section*{Supplementary Note 4 - Proof of Remark \ref{remark1}}\label{sec:proofremark}

We here  present an extension to Theorem \ref{thm1} to the case when several parameters in the ansatz $U(\thv)$ are correlated. Here, by correlated, we mean they are fixed to be equal to each other~\cite{volkoff2021large}. Note that this is in contrast to the previously analyzed cases where we assumed that all parameters $\{\theta_{l m}\}_{l m}$ are independent. Specifically, Remark \ref{remark1} provides an upper bound on the partial derivative of the cost function with respect to a parameter that is degenerate in $\thv$. 

\setcounter{remark}{0}

\begin{remark}[Degenerate parameters]\label{rem:degparamsupp}
   Consider the ansatz defined in Eqs.~\eqref{eq:ansatzSupp}. Suppose there is a subset $G_{st}$ of the set $\{\theta_{l m}\}$ in this ansatz such that $G_{st}$ consists of $g$ parameters that are degenerate: 
\begin{equation}
    G_{st} = \big\{  \theta_{l m} \;|\; \theta_{l m}=\theta_{st} \big\}
\end{equation}        
Here, $\theta_{st}$ denotes the parameter in $G_{st}$ for which $\|\vec{\eta}_{lm} \|_1$ takes the largest value in the set. ($\theta_{st}$ can also be thought of as a reference parameter to which all other parameters are set equal in value.) Then the partial derivative of the noisy cost with respect to $\theta_{st}$ is bounded as

\begin{equation}\label{eqnDegenerateResultSupp}
    |\partial_{st} \widetilde{C}| \leq  \sqrt{8\ln{2}}\, g  N_O \|\vec{\omega}\|_\infty \big\|\vec{\eta}_{st} \big\|_1 \, n^{1/2} q^{cL+1} \,,
\end{equation} 
at all points in the cost landscape.
\end{remark}

\begin{proof}
Using arguments similar to those in Supplementary Note \ref{sec:thm1proof-supp}, we get
\begin{align}
    |\partial_{st} \widetilde{C}| &=
    \sum_{\theta_{hg} \in\, G_{st}} \big|\mathrm{Tr}[O\, \partial_{hg}\rho_L]\big| \\
    &\leq \sum_{\theta_{hg} \in\, G_{st}} \sqrt{8\ln{2}}\, N_O \|\vec{\omega}\|_\infty \big\|\vec{\eta}_{hg} \big\|_1 \,n^{1/2} q^{cL+1}
\end{align}
where the inequality was obtained from  Eq.~\eqref{eq:suppthmfinal}. Since there are $g$ terms in the summation, we have
\begin{equation}\label{eq:remark-dp-supp}
    |\partial_{st} \widetilde{C}| \leq  \sqrt{8\ln{2}}\, g N_O \|\vec{\omega}\|_\infty \big\|\vec{\eta}_{st} \big\|_1 \, n^{1/2} q^{cL+1} \,.
\end{equation}
\end{proof}

We note that the proof of Remark \ref{rem:degparamsupp} can be trivially generalized to the case when the parameters in $G_{st}$ are linear functions of the reference parameter.

\section*{Supplementary Note 5 - Proof of Remark \ref{remark:extension_noisemodel}}\label{sec:remark2-supp}

\begin{remark}[Extensions to the noise model]
   We can extend our noise model to include additional non-local noise models and obtain the same scaling results. First, we may consider global (unital) Pauli noise $\mathcal{P}$ whose action on $n$-qubit Pauli string $\sigma_n \in \{\emph{\id}, X, Y, Z \}^{\otimes n}$ can be written 
   \begin{equation}\label{eq:globalpaulisupp}
       \mathcal{P}(\sigma_n) = q_{\sigma_n}\sigma_n
   \end{equation}
   where $-1\leq q_{\sigma_n} \leq 1$ for all $\sigma_n$, and $q_{\emph{\id}^{\otimes n}} = 1$. Second, we can consider correlated coherent noise across multiple qubits of the form
   \begin{equation}\label{eq:correlatednoisesupp}
       \mathcal{V}(\rho) = V\rho V^\dag
   \end{equation}
   where $VV^\dag=V^\dag V=\emph{\id}^{\otimes n}$
   . We can then consider a modification of our noisy cost function \eqref{eq:noisycostsupp} as
   \begin{equation}
       \widetilde{C}\mapsto\widetilde{C}' = \mathrm{Tr}\Big[ O\; \big(\mathcal{N}_L\circ \mathcal{U}_L(\vec{\theta}_L) \circ \mathcal{N}_{L-1} \circ \cdots \circ \mathcal{U}_2(\vec{\theta}_2) \circ \mathcal{N}_{1} \circ \mathcal{U}_1(\vec{\theta}_1) \circ \mathcal{N}_0 \big) (\rho_0) \Big]
   \end{equation}
   with $\mathcal{N}_i = \mathcal{V}_i\circ\mathcal{P}_i\circ\mathcal{N}$ for all $i \in [0,L]$, where $\mathcal{P}_i$ and $\mathcal{V}_i$ are specific instances of global Pauli noise and correlated noise of the form of \eqref{eq:globalpaulisupp} and \eqref{eq:correlatednoisesupp} respectively. Under such a modification, the statements of Lemma \ref{lemma1}, Theorem \ref{thm1} and Corollary \ref{cor1} still remain valid.
\end{remark}

\begin{proof}
We can absorb the $\mathcal{V}_i$ channels into the parameterized unitaries and write
\begin{equation}
       \widetilde{C}' = \mathrm{Tr}\Big[ V_L^\dag OV_L\; \big(\mathcal{N}'_L\circ \mathcal{Y}_L(\vec{\theta}_L) \circ \mathcal{N}'_{L-1} \circ \cdots \circ \mathcal{Y}_2(\vec{\theta}_2) \circ \mathcal{N}'_{1} \circ \mathcal{Y}_1(\vec{\theta}_1) \circ \mathcal{N}'_0 \big) (\rho_0) \Big]
\end{equation}
where $\mathcal{Y}_i(\vec{\theta}_i) = \mathcal{U}_i(\vec{\theta}_i)\circ V_{i-1}$ and  $\mathcal{N}'_i = \mathcal{P}_i\circ\mathcal{N}$ for all $i \in [0,L]$. For any operator of the form \eqref{eq:paulidecomposition_supp}, the effect of noise channel $\mathcal{N}_j$ is to map Pauli coefficients as
\begin{equation}
    \lambda_i \xrightarrow[]{\,\mathcal{N}_j\,} q_X^{x(i)}q_Y^{y(i)}q_Z^{z(i)} \lambda_i \,,
\end{equation}
for all $i \in [1,2^n-1]$, $j \in [0,L]$, where we can write $q=\max\{|q_X|,|q_Y|,|q_Z|\}<1$. In addition, the results of Supplementary Lemma \ref{suplem:noise+unitaries} are unchanged under the map $O \mapsto V_L^\dag OV_L$. Thus, the above proofs proceed the same under such an extended noise model.
\end{proof}

\section*{Supplementary Note 6 - Proof of Corollaries \ref{cor:qaoa} and \ref{cor:ucc} }\label{sec:cor2sup}

We first provide a proof of Corollary \ref{cor:qaoa}.  We start by recalling that in the QAOA one sequentially alternates the action of two unitaries as    
\begin{equation}\label{eq:QAOAsupp}
    U(\vec{\gamma},\vec{\beta})= e^{-i \beta_{p} H_M}  e^{-i \gamma_{p} H_P}\cdots e^{-i \beta_{1} H_M}  e^{-i \gamma_{1} H_P} \,,
\end{equation}
where $H_P$ and $H_M$ are the so-called problem and mixer Hamiltonian, respectively. We define $N_P\,(N_M)$ the number of terms in the Pauli decompositions of $H_P\,(H_M)$.

\setcounter{corollary}{1}
\begin{corollary}[Example: QAOA]\label{cor2sup}
Consider the QAOA with $2p$ trainable parameters, as defined in Eq.~\eqref{eq:QAOAsupp}. Suppose that the implementation of unitaries corresponding to the problem Hamiltonian $H_P$ and the mixer Hamiltonian $H_M$ require $k_P$- and $k_M$-depth circuits, respectively.
If local Pauli noise of the form in Eq.~\eqref{eq:noisemodelSupp} with noise parameter $q$ acts before and after each layer of native gates, then we have
    \begin{align}
        |\partial_{\beta_l} \widetilde{C}| \, &\leq
         \sqrt{8\ln{2}}\, {g_{l,P}}N_P \|\vec{\omega}\|_\infty \big\|\vec{\eta}_{P} \big\|_1  n^{1/2} q^{c(k_P+k_M)p+1}\,, \label{eq:qaoabound1supp}\\
        |\partial_{\gamma_l} \widetilde{C}| \, &\leq  \sqrt{8\ln{2}}\, {g_{l,M}}N_P   \|\vec{\omega} \|_\infty \big\|\vec{\eta}_{M} \big\|_1 \,n^{1/2}  q^{c(k_P+k_M)p+1}\,, \label{eq:qaoabound2supp}
    \end{align}
    for any choice of parameters $\beta_l,\gamma_l$, and where $O=H_P$ in Eq.~\eqref{eq:OSupp}. Here $b_{l,P}$ and $b_{l,M}$ are respectively the number of native gates parameterized by $\beta_l$ and $\gamma_l$ according to the compilation.
\end{corollary}

\begin{proof}
We now treat each layer of native hardware gates as a unitary layer as in Fig.~\ref{fig:suppcircuit}, which gives $L = (k_P+k_M)p$.
In Eq.~\eqref{eq:remark-dp-supp} we have $\|\vec{\eta}_{\beta_l} \|_1 \leq \|\vec{\eta}_{P} \|_1$, $\|\vec{\eta}_{\gamma_l} \|_1 \leq \|\vec{\eta}_{M} \|_1$, assuming Trotterization.
Then Corollary \ref{cor2sup} follows by invoking Remark \ref{rem:degparamsupp}. 
\end{proof}


Now let us provide a proof for Corollary \ref{cor:ucc}. We recall that the UCC ansatz can be expressed as
\begin{equation}\label{eq:uccsupp}
        U(\vec{\theta}) =  \prod_{l m} U_{l m}(\theta_{l m})= \prod_{l m} e^{i  \theta_{l m} \sum_k   \mu^k_{l m}  \sigma^k_n}, 
\end{equation}
where $\mu^k_{l m} \in \{0,\pm1\}$, and where $\theta_{l m}$ are the coupled cluster amplitudes. Moreover, we denote $\widehat{N}_{l m} =\|\vec{\mu}_{l m}\|_1$ as the number of non-zero elements in  $\sum_k  \mu^k_{l m}  \sigma^k_n$. 

\begin{corollary}[Example: UCC] \label{cor:uccsupp}
Let $H$ denote a molecular Hamiltonian of a system of $M_e$ electrons. Consider the UCC ansatz as defined in Eq.~\eqref{eq:uccsupp}.  If local Pauli noise of the form in Eq.~\eqref{eq:noisemodelSupp} with noise parameter $q$ acts before and after every $U_{l m}(\theta_{l m})$ in Eq.~\eqref{eq:uccsupp}, then we have
\begin{align}\label{eq:ucc-corsupp}
 |\partial_{\theta_{l m}} \widetilde{C}| \leq \sqrt{8\ln{2}}\,\widehat{N}_{l m} N_H \Vert \vec{\omega}\Vert_{\infty} \,n^{1/2} q^{cL+1}, 
\end{align}
for any coupled cluster amplitude $\theta_{l m}$,  and where $O=H$ in  Eq.~\eqref{eq:noisycostsupp}. 
\end{corollary}

\begin{proof}
Using the first-order Troterrization, the UCC ansatz can be represented as follows: 
\begin{align}\label{eq:trotterized-ucc-supp}
    U(\vec{\theta}) = \prod_{l m} \prod_k  e^{i  \theta_{l m}     \mu^k_{l m}  \sigma^k_n},
\end{align}
which is in the form of an ansatz that has correlated parameters. Then from Remark \ref{rem:degparamsupp} it follows that 
\begin{align}
 |\partial_{\theta_{l m}} \widetilde{C}| \leq \sqrt{8\ln{2}}\,\widehat{N}_{l m} N_H \Vert \vec{\omega}\Vert_{\infty} \,n^{1/2} q^{cL+1}, 
\end{align}
where we used the fact that in Eq.~\eqref{eq:remark-dp-supp} $g=\widehat{N}_{l m}$, $\Vert \vec{\eta}_{st}\Vert_{1} =1$ for the UCC ansatz as in Eq.~\eqref{eq:trotterized-ucc-supp}. 
\end{proof}

\section*{Supplementary Note 7 - Proof of Remark \ref{remark:QML}}\label{sec:remark:QML-supp}

Suppose our noiseless cost function is instead 
\begin{equation} \label{eq:C_train}
    C_{\textrm{train}}= \sum_i \Tr[O_i U(\thv)\rho_i U\ad(\thv)]
\end{equation}
for a training set of data encoded in states $\{ \rho_i\}$ and set of operators $\{ O_i\}$ each of the form \eqref{eq:OSupp}. Our result then generalizes to
\begin{equation}
    \big|\partial_{l m} \widetilde{C}_{\textrm{train}}\big| \leq \sqrt{8\ln{2}}\, \Big(\sum_i N_{O_i} \big\|\vec{\omega}_i\big\|_\infty \Big) \big\|\vec{\eta}_{lm}\big\|_1 n^{1/2} q^{cL+1} \,.
\end{equation}
\begin{proof}
We can denote each term in the sum \eqref{eq:C_train} as $C_i=\Tr[O_i U(\thv)\rho_i U\ad(\thv)]$ such that 
\begin{equation}
    C_{\textrm{train}}= \sum_i C_i,\quad \widetilde{C}_{\textrm{train}} = \sum_i \widetilde{C}_i\,.
\end{equation}
Then, we can simply write
\begin{align}
    \big|\partial_{l m} \widetilde{C}_\textrm{train}\big| &= \Big| \sum_i \partial_{l m} \widetilde{C}_i \Big| \\
    &\leq \sum_i \big|\partial_{l m} \widetilde{C}_i \big|\,,
\end{align}
where each $\big|\partial_{l m} \widetilde{C}_i \big|$ can be bounded by Theorem \ref{thm1si} giving the result as required.
\end{proof}

\section*{Supplementary Note 8 - Proof of Proposition \ref{prop1} }\label{sec:prop1si}

In this supplementary note we provide a proof of Proposition \ref{prop1}, which we now recall for convenience.

\setcounter{proposition}{0}
\begin{proposition}[Measurement noise]\label{prop1si}
    Consider the expansion of the observable $O$ as a sum of Pauli strings, as in Eq.~\eqref{eq:HamiltonianDef}. Let $w$ denote the minimum weight of these strings, where the weight is defined as the number of non-identity elements for a given string. In addition to the noise process considered in Fig.~\ref{fig:suppcircuit}, suppose there is also measurement noise consisting of a tensor product of local bit-flip channels with  bit-flip probability $(1-q_M)/2$. Then we have
    \begin{equation}\label{eq:prop-b13}
        \left\vert\widetilde{C} - \frac{1}{2^n}\Tr\,O\right\vert \leq q_M^{w}\, G(n)\,,
    \end{equation}
    and
    \begin{equation}
        |\partial_{l m} \widetilde{C}| \leq q_M^{w} F(n) \,,
    \end{equation}
    where $G(n)$ and $F(n)$ are defined in Lemma~\ref{lemma1} and Theorem~\ref{thm1}, respectively.
\end{proposition}

\begin{proof}
We prove in detail the proposition about the gradient of the cost function. The proposition about the cost function is derived in an analogous manner.

As a model of measurement noise we consider a classical bit-flip channel applied to every qubit, such that the standard POVM elements get replaced by:
\begin{align}
    P_0 = \dya{0} &\rightarrow \widetilde{P}_0 = p_{00} \dya{0} + p_{01} \dya{1} \\
    P_1 = \dya{1} &\rightarrow \widetilde{P}_1 = p_{10} \dya{0} + p_{11} \dya{1} ,
\end{align}
where $p_{00} + p_{01} = 1$ and $p_{10} + p_{11} = 1$. Furthermore, we take this channel to be unital, such that $\widetilde{P}_0 + \widetilde{P}_1 = (p_{00} + p_{10}) P_0 + (p_{01} + p_{11}) P_1 = P_0 + P_1$ giving $p_{00} + p_{10} = 1$ and $p_{01} + p_{11} = 1$. Thus, there is only one free parameter $q_M$, and we set $p_{00} = p_{11} = \frac{1 + q_M}{2}$, $p_{01} = p_{10} = \frac{1 - q_M}{2}$. Note that without loss of generality we can assume $p_{00}, p_{11} > 1/2$, and hence $q_M \geq 0$. Overall:
\begin{align}
    P_0 = \dya{0} &\rightarrow \widetilde{P}_0 = \frac{1+q_M}{2} \dya{0} + \frac{1-q_M}{2} \dya{1} \\
    P_1 = \dya{1} &\rightarrow \widetilde{P}_1 = \frac{1-q_M}{2} \dya{0} + \frac{1+q_M}{2} \dya{1} \,.
\end{align}
The equivalence between this classical channel and a quantum bit-flip channel is seen by writing $P_0 = \frac{\id + Z}{2}$ and $P_1 = \frac{\id - Z}{2}$, such that the bit-flip channel is equivalent to a transformation of the Pauli $Z$ operator: ${Z}' = q_M Z$. This corresponds to the effect of a bit-flip channel $\mathcal{N}(\rho) = \frac{1 + q_M}{2} \rho + \frac{1 - q_M}{2} X\rho X$. 

The reasoning so far only applies to measurements in the $Z$ basis. However, in our model we do not consider a standard projective measurement, but the expectation value with respect to a general Hermitian operator. This assumes the capability of performing measurements in any basis. If we assume that the classical bit-flip acts independently of the basis we choose to measure in, then we see that the corresponding quantum channel must be a depolarizing channel such that
\begin{equation}\label{eq:actionmpn}
    \mathcal{D}_M(\sigma)=q_M\sigma\,,
\end{equation}
where $\sigma$ is any single-qubit Pauli operator. An alternative realistic assumption that also leads to \eqref{eq:actionmpn} is that the quantum computer can only measure in the computational basis, and so one implements measurements in general bases by applying an extra layer of (noisy) one-qubit rotations before measurement. We thus proceed to model measurement noise as a tensor product of such local depolarizing channels applied prior to measurement and denote the overall channel as $\mathcal{N}_M$. From Eq.~\eqref{eq:actionmpn} we have that 
\begin{equation}\label{eq:actionmpn2}
    \mathcal{N}_M(O)=\sum_i \omega^i\mathcal{N}_M(\sigma_n^i)=\widetilde{\vec{\omega}}\cdot \vec{\sigma}_n\,,
\end{equation}
where $\widetilde{\vec{\omega}}$ is a vector of elements $\widetilde{\omega}^i=q_M^{w(i)} \omega^i$, and where $w(i)=x(i) + y(i) + z(i)$ is the weight of the Pauli string. Here we recall that we have respectively defined $x(i)$, $y(i)$ and $z(i)$ as the number of Pauli operators $X$, $Y$, and $Z$ in the $i$-th Pauli string. 
Let us now write the noisy cost function partial derivative as:
\begin{align}
    \partial_{l m} \widetilde{C} &= \text{Tr} \left[\mathcal{N}_M(O)\partial_{l m}\rho_L\right] \,.
\end{align}
Proceeding in the same way as in the proof of Theorem \ref{thm1si}, we write
\begin{align}\label{eq:Hölder_supp2}
    |\partial_{l m} \widetilde{C}| \leq \left\|\NC\circ\mathcal{N}_M(O)\right\|_\infty\, \left\|\mathcal{W}_{+}\big( \partial_{l m}\, {\rho}_{l-} \big)\right\|_1\,.
\end{align}
The first term can be bounded as
\begin{align}
    \left\|\NC\circ\mathcal{N}_M(O)\right\|_\infty &= \left\|\NC\left( \widetilde{\vec{\omega}}\cdot \vec{\sigma}_n \right)\right\|_\infty \\
    &\leq N_O \max_i\big\|\mathcal{N}(\widetilde{\omega}^i \sigma^i_n)\big\|_\infty \\
    &\leq N_O\, q \min_{w(i)}  (q^{w(i)}_M) \|\vec{\omega}\|_{\infty} \\
    &= N_O\, q\, q^{w}_M \|\vec{\omega}\|_{\infty}\,,
\end{align}
where the first inequality is due to the triangle inequality, the second equality comes from noting that the maximization over the product of positive terms is bounded by the product of maximizations, and the final equality is an adpotion of the notations in the statement of the proposition. We see this result is identical to that in Theorem \ref{thm1si}, aside from an extra factor $q_M^{w}$.

The second term in \eqref{eq:Hölder_supp2} is bounded in the proof of Theorem \ref{thm1si} as
\begin{align}
    \left\| \mathcal{W}_{+}\big( \partial_{l m}\, {\rho}_{l-} \big)\right\|_1 \leq \sqrt{8\ln{2}} \,   \big\|\vec{\eta}_{lm}\big\|_1\, n^{1/2} q^{cL}\,. 
    \end{align}
Putting the two parts together we obtain 
\begin{align}
    \big|\partial_{l m} \widetilde{C}\big| &\leq \sqrt{8\ln{2}}\, N_O \big\|\vec{\omega}\big\|_\infty \|\vec{\eta}_{lm}\big\|_1 n^{1/2}\, q_M^{w}\, q^{cL+1} \\
    &= q_M^{w} F(n)\,,
\end{align}
as required.
\bigskip

Now let us prove the complimentary result for the cost function magnitude. Following the proof of Lemma \ref{lemma3-supp} we write
\begin{align}
    \left\vert \widetilde{C} - \frac{1}{2^n} \Tr[O]\right\vert &= \bigg\vert  \Tr\big[\NC\circ\NC_M(O)\,{\vec{{a}}}^{(L)}\cdot \vec{\sigma}_n \big]\bigg\vert \\
    &\leq  \big\|\NC\circ\NC_M(O)\big\|_{\infty} \big\|{\vec{{a}}}^{(L)}\cdot \vec{\sigma}_n \big\|_1 \\
    &\leq q_M^{w}\, N_O \big\|\vec{\omega}\big\|_\infty \sqrt{2 \ln{2} \cdot n}\, q^{cL+1} \,.
\end{align}
Thus we can write
\begin{align}
    \left\vert\widetilde{C} - \frac{1}{2^n} \Tr[O]\right\vert \leq q_M^{w}\, G(n)\,,
\end{align}
where $G(n) = \sqrt{2 \ln{2} }\, N_O \big\|\vec{\omega}\big\|_\infty n^{1/2} q^{cL+1}$ is the concentration factor in Lemma \ref{lemma3-supp}. Hence observables with $w \in \Omega(n)$ will suffer from an exponential decay in $n$ of the cost function and its gradient, independent of circuit depth.
\end{proof}

\end{document}